\documentclass[12pt]{article}
\usepackage{amsmath, amsthm, amssymb, amsfonts, enumerate,bm,url,eufrak}
\usepackage[figuresright]{rotating}
\usepackage{lscape}
\usepackage{multirow}
\usepackage{natbib}
\usepackage[table]{xcolor} 

\begin{document}
\newcommand{\abs}[1]{\left\vert#1\right\vert}
\newcommand{\set}[1]{\left\{#1\right\}}
\newcommand{\eps}{\varepsilon}
\newcommand{\To}{\rightarrow}
\newcommand{\inv}{^{-1}}
\newcommand{\ihat}{\hat{\imath}}
\newcommand{\var}{\mbox{Var}}
\newcommand{\sd}{\mbox{SD}}
\newcommand{\cov}{\mbox{Cov}}
\newcommand{\f}{\frac}
\newcommand{\fI}[1]{\frac{1}{#1}}
\newcommand{\what}[1]{\widehat{#1}}
\newcommand{\hhat}[1]{\what{\what{#1}}}
\newcommand{\wtilde}[1]{\widetilde{#1}}
\newcommand{\bdot}{\bm{\cdot}}
\newcommand{\Th}{\theta}
\newcommand{\qmq}[1]{\quad\mbox{#1}\quad}
\newcommand{\qm}[1]{\quad\mbox{#1}}
\newcommand{\mq}[1]{\mbox{#1}\quad}
\newcommand{\tr}{\mbox{tr}}
\newcommand{\logit}{\mbox{logit}}
\newcommand{\noi}{\noindent}
\newcommand{\bni}{\bigskip\noindent}
\newcommand{\bul}{$\bullet$ }
\newcommand{\bias}{\mbox{bias}}
\newcommand{\conv}{\mbox{conv}}
\newcommand{\spn}{\mbox{span}}
\newcommand{\colspace}{\mbox{colspace}}
\newcommand{\mA}{\mathcal{A}}
\newcommand{\mF}{\mathcal{F}}
\newcommand{\mH}{\mathcal{H}}
\newcommand{\mI}{\mathcal{I}}
\newcommand{\mJ}{\mathcal{J}}
\newcommand{\mN}{\mathcal{N}}
\newcommand{\mR}{\mathcal{R}}
\newcommand{\mT}{\mathcal{T}}
\newcommand{\mX}{\mathcal{X}}
\newcommand{\mS}{\mathcal{S}}
\newcommand{\bbR}{\mathbb{R}}
\newcommand{\fweI}{\mbox{$k_1$-FWER$_1$}}
\newcommand{\fweII}{\mbox{$k_2$-FWER$_2$}}
\newcommand{\fdp}{\mbox{FDP}}
\newcommand{\fnp}{\mbox{FNP}}
\newcommand{\gfdp}{\mbox{$\gamma_1$-FDP}}
\newcommand{\gfnp}{\mbox{$\gamma_2$-FNP}}
\newcommand{\vphi}{\varphi}
\newcommand{\oN}{\overline{N}}

\newtheorem{theorem}{Theorem}[section]
\newtheorem{corollary}{Corollary}[section]
\newtheorem{conjecture}{Conjecture}[section]
\newtheorem{proposition}{Proposition}[section]
\newtheorem{lemma}{Lemma}[section]
\newtheorem{definition}{Definition}[section]
\newtheorem{example}{Example}[section]
\newtheorem{remark}{Remark}[section]

\title{{\bf\Large Multiple Hypothesis  Tests  Controlling Generalized Error Rates for Sequential Data}}

\author{\textsc{Jay Bartroff\footnote{email: \textsf{bartroff@usc.edu}. This work was partially supported by National Science Foundation grant DMS-1310127 and National Institutes of Health grant R01 GM068968.}}\\
\small{Department of Mathematics, University of Southern California, Los Angeles, California, USA}}  

\date{}  
\maketitle

\abstract{The $\gamma$-FDP and $k$-FWER multiple testing error metrics, which are tail probabilities of the respective error statistics, have become popular recently as alternatives to the FDR and FWER. We propose general and flexible stepup and stepdown procedures for testing multiple hypotheses about sequential (or streaming) data that simultaneously control both the type I and II versions of $\gamma$-FDP, or $k$-FWER. The error control holds regardless of the dependence between data streams, which may be of arbitrary size and shape.  All that is needed is a test statistic for
each data stream that controls the conventional type I and II error probabilities, and no information or assumptions are required about the joint distribution of the statistics or data
streams. The procedures can be used with sequential, group sequential, truncated, or other
sampling schemes. We give recommendations for the procedures' implementation including closed-form expressions for the needed critical values in some commonly-encountered testing situations. The proposed sequential procedures are compared with each other and with comparable fixed sample size procedures in the context of strongly positively correlated Gaussian data streams. For this setting we conclude that both the stepup and stepdown sequential procedures provide substantial savings over the fixed sample procedures in terms of expected sample size, and the stepup procedure performs slightly but consistently better than the stepdown for $\gamma$-FDP control, with the relationship reversed for $k$-FWER control.}

\section{Introduction and Summary}\label{sec:intro}

Driven in part by modern applications involving high-dimensional models or the need for many comparisons in areas such as high-throughput gene and protein expression data, brain imaging, and astrophysics, there has been much interest and innovation during recent decades in statistical methodology involving multiple testing error rates which are less stringent than the classical \textit{familywise error rate (FWER)}, the probability of rejecting at least one true null hypothesis. \citet{Hommel88}  proposed the $k$-FWER, the probability of rejecting at least $k\ge 1$ true null hypotheses, and this was independently proposed later by \citet{Lehmann05b}.  \citet{Benjamini95} proposed the \textit{false discovery rate (FDR)}, the expectation of the \textit{false discovery proportion (FDP)}, the latter being the proportion of rejected null hypotheses that are true. As a generalization of the FDR,  \citet{Lehmann05b} proposed using the probability that the FDP exceeds a fixed value $\gamma\in[0,1)$, which has come to be known as the $\gamma$-FDP. Recently, Guo, He, and Sarkar~\citeyearpar{Guo14} proposed a further generalization of the $\gamma$-FDP.  Each of these authors also supplied procedures to control the respective generalized error rates under various dependence assumptions on the data, ranging from independence to positive regression dependency on subsets \citep{Benjamini01} to no assumptions at all. Many other authors have also provided innovative new procedures and theory surrounding these generalized error rates and we do not attempt to summarize this large and growing literature here, but instead refer the reader to \citet{Guo14} and the references therein.

All of the references mentioned above take as their starting point a set of valid $p$-values corresponding to fixed sample size tests for the list of null hypotheses of interest. However, in some areas of application the data in a multiple testing setup does not naturally occur in a fixed sample but rather arrives sequentially (or in groups) in time, referred to as ``streaming'' data in some applications. An obvious example is in biomedical clinical trials with multiple endpoints or arms \citep[e.g.,][Chapter~15]{Jennison00}, but  others areas with naturally sequential data abound including certain high-throughput sequencing technologies \citep{Salzman11,Jiang12}, multi-channel changepoint detection \citep{Tartakovsky03}, biosurveillance \citep{Mei10}, acceptance sampling with multiple criteria \citep{Baillie87}, financial data \citep{Lai08}, and some agricultural studies \citep{Clements14}. Only recently have general and flexible multiple testing procedures suited for the particular needs of sequential  data been proposed in literature. \citet{Bartroff10e} proposed a sequential version of Holm's \citeyearpar{Holm79} FWER-controlling procedure.  \citet{De12,De12b} proposed procedures controlling both the type~I and II FWER under the restriction that all data streams be sampled until accept/reject decisions can be reached for all null hypotheses simultaneously, and \citet{Bartroff14b} proposed a procedure lifting this restriction. Like Holm's procedure, each of these sequential procedures mentioned so far has guaranteed FWER control under arbitrary dependence of data streams. \citet{Bartroff14c} proposed an analogous procedure controlling FDR and its type~II analog, the false nondiscovery rate, on sequential data.

The purpose of this paper is to provide general and flexible procedures for controlling $k$-FWER and $\gamma$-FDP on sequential data. By ``general and flexible'' we mean procedures that can test $J\ge 2$ arbitrary null hypotheses $H^{(1)},\ldots,H^{(J)}$ about $J$ data streams
 \begin{align}
\mbox{Data stream $1$:}\quad&X_1^{(1)},X_2^{(1)},\ldots\nonumber\\
\mbox{Data stream $2$:}\quad&X_1^{(2)},X_2^{(2)},\ldots\label{streams}\\
\vdots&\nonumber\\
\mbox{Data stream $J$:}\quad&X_1^{(J)},X_2^{(J)},\ldots,\nonumber
\end{align} respectively, of arbitrary size, shape, and dependence. In particular, each data point~$X_n^{(j)}$ may itself be the vector of observations from the $n$th group, corresponding to group sequential sampling.  This setup is formalized below. 
In particular we define stepdown and stepup procedures that require only arbitrary sequential test statistics $\Lambda^{(j)}(n)=\Lambda^{(j)}(n)(X_1^{(j)},\ldots, X_n^{(j)})$ for each stream $j=1,\ldots,J$ that control the conventional type~I and II error probabilities for each individual null hypothesis~$H^{(j)}$, and combine them to give a sequential multiple testing procedure, i.e., a collection of $J$ sequential stopping and decision rules for each data stream, that  controls $k$-FWER or $\gamma$-FDP at a prescribed level under arbitrary dependence structure between  data streams. In this regard our procedures can be viewed as extensions to the sequential realm of the procedures of \citet{Lehmann05b} and \citet{Romano06b,Romano06} who accomplished this in the  fixed sample setup. Indeed,  our approach owes much to the work of these authors and, in particular, we utilize the same stepup and stepdown values as those papers.  \citet{Sarkar07,Sarkar08} and \citet{Guo14} have furthered the work of these authors by developing stepup and stepdown fixed sample size procedures which utilize the joint null distribution of the $p$-values and, in some cases, dominate previously proposed procedures while controlling generalized error rates. While we expect that these innovations by Sarkar and his coauthors can similarly be extended to the sequential domain, since our goal here is to propose procedures that don't require knowledge (or modeling) of joint distributions, we have not pursued those extensions here.

An additional aspect of our approach is that our procedures may be able to simultaneously control  both the type~I and II versions of the generalized error metrics at prescribed values, which is a possibility opened up by the sequential setting considered. As with single hypothesis testing, if the prescribed type~II error rate (or equivalently, power) is not well-motivated, then its strict control can be dispensed with or used as a surrogate for other operating characteristics of interest, such as average sample size. For this reason and others discussed in Section~\ref{sec:rej}, there we provide versions of the procedures that control the type~I generalized error rate but not necessarily the type~II version, and these procedures can be used with arbitrary acceptance rules for the null hypotheses.  

Regarding our sequential setup, we remark that in order to sequentially test $J\ge 2$ null hypotheses, one could simply apply $J$ chosen sequential stopping rules to the data streams, calculate the (appropriately adjusted) $p$-values upon stopping, and then apply a fixed-sample procedure to the $p$-values. However, this ``naive'' method will in general be inefficient compared to the procedures proposed herein since the stopping rules do not explicitly take the multiple testing error metric into account. Moreover, the naive method will not in general control both the type~I and II multiple testing error rates which, even when this feature is not a priority of the statistician as mentioned in the previous paragraph, means that the relationship between the naive method's stopping rule and its power is not well understood, unlike the proposed procedures. Nonetheless, our approach below could be applied by taking the test statistics~$\Lambda^{(j)}(n)$ to be sequential $p$-values, making it look more like the fixed sample size procedures. Instead we have chosen to use arbitrary sequential test statistics to maintain generality and to make the resulting procedure more user-friendly,  given that other types of test statistics like log-likelihood ratios (or simple functions thereof) are much more commonly used with sequential data than sequential $p$-values. This is perhaps due to the complexity and non-uniqueness of sequential $p$-values in all but the simplest cases; see \citet[][Chapters~8.4 and 9]{Jennison00}.

The remainder of this paper is organized as follows.  After introducing the notation and setup in Section~\ref{sec:setup}, in Section~\ref{sec:D} we define a ``generic'' sequential stepdown procedure that accepts arbitrary stepdown values~\eqref{UD.values}, special cases of which are given in Sections~\ref{sec:DP.procs} and \ref{Dk.procs} which control type~I and II $\gamma$-FDR and $k$-FWER, respectively. An analogous development of stepup procedures in given in Section~\ref{sec:U}. Section~\ref{sec:rej} gives versions of these procedures with only explicit rejection rules for use when the type~II error rate is not well motivated or there is a restriction on maximum sample size. In Section~\ref{sec:imp} we give recommendations for implementing the procedures by reviewing how to implement sequential single-hypothesis tests in some commonly-encountered situations, and we give closed-form expressions for the needed critical values in Theorem~\ref{thm:simple}. Section~\ref{sec:group.seq} discusses how to implement group sequential sampling. Section~\ref{sec:sim} contains the results of a numerical study comparing the proposed stepup, stepdown, and comparable fixed sample size procedures in a setting of strongly positively correlated Gaussian data streams. In Section~\ref{sec:conc} we summarize our recommendations. All proofs are delayed until Section~\ref{sec:proofs}.

\section{Setup}\label{sec:setup}
\subsection{Data Streams, Hypotheses, and Error Metrics}
 Assume that there are $J\ge 2$ data streams \eqref{streams}. In general we make no assumptions about the dimension  of the sequentially-observed data $X_n^{(j)}$, which may themselves be vectors of varying size, nor about the dependence structure of within-stream data $X_n^{(j)}, X_m^{(j)}$ or between-stream data $X_n^{(j)}, X_m^{(j')}$ ($j\ne j'$). In particular there can be arbitrary ``overlap'' between data streams, an extreme case being that all the data streams are the same, which is equivalent to testing multiple hypotheses about a single data source. For any positive integer $j$ let $[j]=\{1,\ldots,j\}$. For each data stream, indexed by $j\in[J]$, assume that there is a parameter vector $\theta^{(j)}\in\Theta^{(j)}$ determining that distribution of the stream $X_1^{(j)}, X_2^{(j)},\ldots$, and it is desired to test a null hypothesis $H^{(j)}$ versus the alternative hypothesis $G^{(j)}$, where $H^{(j)}$ and $G^{(j)}$ are disjoint subsets of the parameter space~$\Theta^{(j)}$ containing $\theta^{(j)}$. The null~$H^{(j)}$ is considered \textit{true} if $\theta^{(j)}\in H^{(j)}$, and \textit{false} if $\theta^{(j)}\in G^{(j)}$. The global parameter $\theta=(\theta^{(1)},\ldots,\theta^{(J)})$ is the concatenation of the individual parameters and is contained in the global parameter space $\Theta=\Theta^{(1)}\times\cdots\times \Theta^{(J)}$. Let
\begin{equation}\label{T}
\mathcal{T}(\theta) = \{ j\in[J]: \theta^{(j)} \in H^{(j)}\}
\end{equation}
denote the indices of the true null hypotheses when $\theta$ is the true global parameter,  and 
\begin{equation}\label{F}
\mF(\theta) = \{ j\in[J]: \theta^{(j)} \in G^{(j)}\}
\end{equation}
the indices of the false null hypotheses.

It may appear that the notation~\eqref{streams} for the data streams restricts us to fully-sequential sampling where the streamwise sample sizes may take any value $1,2,\ldots$ \textit{ad infinitum}. However, since the observations $X_n^{(j)}$ themselves may be of arbitrary size and shape, group sequential (and even variable-stage size) sampling fits into this framework. To wit, the $n$th ``observation''~$X_n^{(j)}$ in the $j$th stream may actually be the $n$th group $X_n^{(j)}=(X_{n,1}^{(j)},\ldots,X_{n,\ell}^{(j)})$ of size $\ell$. Moreover, the group size~$\ell$ may vary with $n$ and may even be data-dependent, e.g., determined by some type of adaptive sampling. Similarly, truncated sampling can be implemented for the $j$th stream by defining $X_n^{(j)}=\emptyset$ for all $n>\overline{N}^{(j)}$ for some stream-specific truncation point $\overline{N}^{(j)}$, or globally for all streams by replacing statements like ``for some $n$'' in what follows with ``for some $n\le \overline{N}$,'' for some global truncation point $\overline{N}$. 

The FDP is formally defined as
\begin{equation}\label{FDP.def}
\fdp(\theta)= \begin{cases}
 \frac{\mbox{the number of $H^{(j)}$ rejected, $j\in\mT(\theta)$}}{\mbox{the number of $H^{(j)}$ rejected}},&\mbox{if the denominator is positive,}\\
 0,&\mbox{otherwise.}
\end{cases}
\end{equation} For example, as mentioned above, Benjamini and Hochberg's \citeyearpar{Benjamini95} FDR is the expectation $E_\theta(\fdp(\theta))$ of the FDP.  Since we will consider procedures that simultaneously control both the type~I and type~II versions of the generalized error rates, we also define the type~II analog of FDP, which we call the \textit{false nondiscovery proportion (FNP)}, 
\begin{equation*}
\fnp(\theta)= \begin{cases}
 \frac{\mbox{the number of $H^{(j)}$ accepted, $j\in\mF(\theta)$}}{\mbox{the number of $H^{(j)}$ accepted}},&\mbox{if the denominator is positive,}\\
 0,&\mbox{otherwise.}
\end{cases}
\end{equation*}
With FDP and FNP nailed down, for $\gamma_1,\gamma_2\in[0,1)$  we define
$$\gfdp(\theta)=P_\theta(\fdp(\theta)>\gamma_1)\qmq{and} 
\gfnp(\theta)=P_\theta(\fnp(\theta)>\gamma_2).
$$

Similarly, for $k$-FWER we will distinguish the type~I and II versions by, for $k_1,k_2\in[J]$, defining
\begin{align*}
\fweI(\theta) &= P_\theta(\mbox{at least $k_1$ null hypotheses $H^{(j)}$ rejected, $j\in\mathcal{T}(\theta)$}),\\
\fweII(\theta) &= P_\theta(\mbox{at least $k_2$ null hypotheses $H^{(j)}$ accepted, $j\in\mathcal{F}(\theta)$}).
\end{align*} We will omit the argument $\theta$ from these quantities in what follows when it causes no confusion.

\subsection{Test Statistics and Critical Values}\label{sec:stats}
The building blocks of the sequential procedures defined below are $J$ individual sequential test statistics $\{\Lambda^{(j)}(n)\}_{j\in[J],\; n\ge 1}$, where $\Lambda^{(j)}(n)$ is the statistic for testing $H^{(j)}$ vs.\ $G^{(j)}$ based on the data $X_1^{(j)},X_2^{(j)},\ldots,X_n^{(j)}$ available  from the $j$th stream at time $n$.  For example, $\Lambda^{(j)}(n)$ may be a sequential log likelihood ratio statistic for testing $H^{(j)}$ vs.\ $G^{(j)}$. The stepup and stepdown procedures proposed below will be defined in terms of given constants
\begin{equation}\label{UD.values}
0\le\alpha_1\le\ldots\le\alpha_J\le 1\qmq{and}0\le\beta_1\le\ldots\le\beta_J\le 1,
\end{equation} which we will refer to as \textit{step values}, the $\alpha_j$ corresponding to type~I error control and the $\beta_j$ to type~II. These values will be used in a similar way as in fixed sample size stepdown and stepup procedures, which we review now for comparison with what follows. Based on $p$-values $p^{(j_1)}\le\ldots\le p^{(j_J)}$ with $p^{(j)}$ corresponding to $H^{(j)}$, the stepdown procedure based on constants $\alpha_j$ satisfying \eqref{UD.values} rejects $H^{(j_1)},\ldots, H^{(j_d)}$ where $d=\max\{i\in[J]: p^{(j_{i'})}\le \alpha_{i'}\;\mbox{for all}\; i'\le i\}$ (accepting all nulls if the maximum doesn't exist), whereas the stepup procedure rejects $H^{(j_1)},\ldots, H^{(j_u)}$ where $u=\max\{i\in[J]: p^{(j_i)}\le \alpha_i \}$ (accepting all nulls if the maximum doesn't exist). Note that $d\le u$ so that the stepup procedure rejects at least as many null hypotheses as the corresponding stepdown procedure using the same step values.

Given step values $\{\alpha_j, \beta_j\}_{j\in[J]}$, for each test statistic $\Lambda^{(j)}(n)$ we assume the existence of critical values $\{A_w^{(j)}, B_w^{(j)}\}_{w\in[J]}$ such that
\begin{align}
P_{\theta^{(j)}}(\Lambda^{(j)}(n)\ge B_w^{(j)}\;\mbox{some $n$,}\; \Lambda^{(j)}(n')>A_1^{(j)}\;\mbox{all $n'<n$})&\le \alpha_w\qmq{for all}\theta^{(j)}\in H^{(j)}\label{typeI}\\
 P_{\theta^{(j)}}(\Lambda^{(j)}(n)\le A_w^{(j)}\;\mbox{some $n$,}\; \Lambda^{(j)}(n')<B_1^{(j)}\;\mbox{all $n'<n$})&\le \beta_w\qmq{for all}\theta^{(j)}\in G^{(j)}\label{typeII}
\end{align} for all $w\in[J]$. The critical values $A_1^{(j)}, B_1^{(j)}$ are simply the critical values for the sequential test that samples until $\Lambda^{(j)}(n)\not\in (A_1^{(j)}, B_1^{(j)})$, and this test has type~I and II error probabilities bounded above by $\alpha_1$ and $\beta_1$, respectively.  The values $B_w^{(j)}$, $w\in[J]$, are then such that the similar sequential test with critical values $A_1^{(j)}$ and $B_w^{(j)}$ has type~I error probability $\alpha_w$, which is just a restatement of \eqref{typeI}, with an analogous statement holding for critical values $A_w^{(j)}$ and $B_1^{(j)}$, type~II error probability $\beta_w$, and \eqref{typeII}. The reason that critical values $A_1^{(j)}$ and $B_w^{(j)}$ are considered in \eqref{typeI} for type~I error probability control and not, say, $A_w^{(j)}$ and $B_w^{(j)}$ is that the procedures defined below will sample during the $i$th stage using critical values $A_w^{(j)}$ and $B_{w'}^{(j)}$ for some fixed values $w,w'\in[J]$ determined by the data in the previous stages $1,\ldots,i-1$. The probability that, during the $i$th stage, $\Lambda^{(j)}(n)\ge B_{w'}^{(j)}$  before $\Lambda^{(j)}(n)\le A_w^{(j)}$ will then be bounded above by the corresponding statement with $A_w^{(j)}$ replaced by $A_1^{(j)}$, using the fact that $A_1^{(j)}\le A_w^{(j)}$ by \eqref{AB.mono}, and thus this probability related to \eqref{typeI} after bounding $w'$.  Analogous statements apply regarding bounding the type~II error probability.

In all commonly-encountered testing situations there are standard sequential statistics whose critical values can be chosen that satisfy these error bounds, for any given $\{\alpha_j, \beta_j\}_{j\in[J]}$ \citep[][give examples]{Bartroff14b}. Without loss of generality we assume that, for each $j\in[J]$, 
\begin{gather}
A_1^{(j)}\le A_2^{(j)}\le\ldots\le A_J^{(j)}\le B_J^{(j)}\le B_{J-1}^{(j)}\le\ldots\le B_1^{(j)}, \label{AB.mono}\\
\mbox{$A_w^{(j)}=A_{w+1}^{(j)}$ if and only if  $\beta_w=\beta_{w+1}$,}\label{A.mono}\\
\mbox{$B_w^{(j)}=B_{w+1}^{(j)}$ if and only if  $\alpha_w=\alpha_{w+1}$.}\label{B.mono}
\end{gather}
A simplistic example of how critical values \eqref{AB.mono} are used in our sequential multiple testing procedure will be given in the last two paragraphs of Section~\ref{sec:gen.D}.

The sequential multiple testing procedures proposed below will involve ranking the test statistics associated  with different data streams, which may be on completely different scales in general, so for each stream $j$ we introduce a \textit{standardizing function} $\vphi^{(j)}(\cdot)$ which will be applied to the statistic $\Lambda^{(j)}(n)$ before ranking. The standardizing functions  $\vphi^{(j)}$ can be any increasing functions such that $\vphi^{(j)}(A_w^{(j)})$ and $\vphi^{(j)}(B_w^{(j)})$ do not depend on $j$, and we let
\begin{equation}\label{ab.def}
a_w=\vphi^{(j)}(A_w^{(j)})\qmq{and}b_w=\vphi^{(j)}(B_w^{(j)}),\quad j,w\in[J],
\end{equation} denote these common values.  Given critical values $\{A_w^{(j)},B_w^{(j)}\}_{j,w\in[J]}$ satisfying \eqref{typeI}-\eqref{typeII}, one may choose arbitrary values $\{a_w,b_w\}_{w\in[J]}$ satisfying the same monotonicity conditions as the $\{A_w^{(j)},B_w^{(j)}\}$ according to \eqref{A.mono}-\eqref{B.mono} and then define the standardizing functions $\vphi^{(j)}(\cdot)$ to be increasing, piecewise linear functions satisfying \eqref{ab.def}. For example, if all the $\alpha_w$ are distinct and the $\beta_w$ are distinct then a simple choice for the $\{a_j,b_j\}$ are the integers
$$a_1=-J,\quad a_2=-J+1,\quad \ldots,\quad a_J=-1, \quad b_J=1, \quad b_{J-1}=2,\quad \ldots,\quad b_1=J.$$
In any case, the assumptions on the critical values and standardizing functions imply that the $a_w$ must be nondecreasing and the $b_w$ nonincreasing.
Finally, we denote $\wtilde{\Lambda}^{(j)}(n)= \varphi^{(j)}(\Lambda^{(j)}(n))$ and then \eqref{typeI}-\eqref{typeII} can be written
\begin{align}
P_{\theta^{(j)}}(\wtilde{\Lambda}^{(j)}(n)\ge b_w\;\mbox{some $n$,}\; \wtilde{\Lambda}^{(j)}(n') > a_1\;\mbox{all $n'<n$})&\le \alpha_w\qmq{for all}\theta^{(j)}\in H^{(j)}\label{typeI.stand}\\
 P_{\theta^{(j)}}(\wtilde{\Lambda}^{(j)}(n)\le a_w\;\mbox{some $n$,}\; \wtilde{\Lambda}^{(j)}(n')< b_1\;\mbox{all $n'<n$})&\le \beta_w \qmq{for all}\theta^{(j)}\in G^{(j)},\label{typeII.stand}
\end{align} for all $j,w\in[J]$.

\section{Procedures Controlling Type~I and II Generalized Error Rates}\label{sec:I&II}

\subsection{Stepdown Procedures}\label{sec:D}

\subsubsection{The Generic Sequential Stepdown Procedure}\label{sec:gen.D}

Here we define a generic sequential stepdown procedure, special cases of which will be used to define the type~I and II $k$-FWER and $\gamma$-FDP controlling sequential procedures below.  We assume that  step values $\{\alpha_j, \beta_j\}_{j\in[J]}$ satisfying \eqref{UD.values} are given and that the test statistics and critical values satisfy the assumptions in Section~\ref{sec:stats} with respect to these values.

We will describe the procedure in terms of stages of sampling, between which reject/accept decisions are made. Let $\mJ_i\subseteq[J]$ ($i=1,2,\ldots$) denote the index set of the \emph{active} data streams (i.e., those whose corresponding null hypothesis $H^{(j)}$ has been neither accepted nor rejected yet) at the beginning of the $i$th stage of sampling, and $n_i$ will denote the cumulative sample size of any active test statistic up to and including the $i$th stage. The total number of null hypotheses that have been rejected (resp.\ accepted) at the beginning of the $i$th stage will be denoted by $r_i$ (resp.\ $c_i$). Accordingly, set $\mJ_1=[J]$, $n_0=0$, $r_1=c_1 =0$. Let $|\cdot|$ denote set cardinality. Then the $i$th stage of sampling ($i=1,2,\ldots$) of the \textbf{Generic Sequential Stepdown Procedure} with step values $\{\alpha_j, \beta_j\}_{j\in[J]}$ proceeds as follows.

\begin{enumerate}
\item\label{sample-step} Sample the active streams $\{X_n^{(j)}\}_{j\in\mJ_i, \; n>n_{i-1}}$ until $n$ equals
\begin{equation}\label{cont-samp}n_i=\inf\left\{n>n_{i-1}: \wtilde{\Lambda}^{(j)}(n)\not\in (a_{c_i+1}, b_{r_i+1}) \qmq{for some} j\in\mJ_i \right\}.\end{equation}
\item\label{step:ord} Order the active test statistics
\begin{equation*}
\wtilde{\Lambda}^{(j(n_i,1))}(n_i)\le \wtilde{\Lambda}^{(j(n_i,2))}(n_i)\le \ldots\le \wtilde{\Lambda}^{(j(n_i,|\mJ_i|))}(n_i),\end{equation*} where $j(n_i,\ell)$ denotes the index of the $\ell$th ordered active statistic at the end of stage~$i$.
\item  
\begin{enumerate}
\item\label{rej-step} If the upper boundary in \eqref{cont-samp} has been crossed, that is, if $\wtilde{\Lambda}^{(j)}(n_i)\ge b_{r_i+1}$ for some $j\in\mJ_i$, then reject the $m_i\ge 1$ null hypotheses \begin{equation}\label{Hsrej}
H^{(j(n_i,|\mJ_i|))}, H^{(j(n_i,|\mJ_i|-1))}, \ldots, H^{(j(n_i,|\mJ_i|-m_i+1))}, 
\end{equation}
where 
\begin{equation*}
m_i=\max\left\{m\in [|\mJ_i|]: \wtilde{\Lambda}^{(j(n_i,\ell))}(n_i)\ge b_{r_i+|\mJ_i|-\ell+1}\qmq{for all}\ell=|\mJ_i|-m+1,\ldots,|\mJ_i| \right\},
\end{equation*} 
and set $r_{i+1}=r_i+m_i$. Otherwise set $r_{i+1}=r_i$.

\item\label{acc-step} If the lower boundary in \eqref{cont-samp} was crossed,  that is, if $\wtilde{\Lambda}^{(j)}(n_i)\le a_{c_i+1}$ for some $j\in\mJ_i$, then accept the $m_i'\ge 1$ null hypotheses $$H^{(j(n_i,1))}, H^{(j(n_i,2))}, \ldots, H^{(j(n_i,m_i'))},$$ where 
\begin{equation*}
m_i'=\max\left\{m\in[|\mJ_i|] : \wtilde{\Lambda}^{(j(n_i,\ell))}(n_i)\le a_{c_i+\ell} \qmq{for all}\ell=1,\ldots,m\right\},
\end{equation*} and set $c_{i+1}=c_i+m_i'$. Otherwise set $c_{i+1}=c_i$.
\end{enumerate}

\item\label{stop-step} Stop if there are no remaining active hypotheses, i.e., if $r_{i+1}+c_{i+1}=J$.  Otherwise, let $\mJ_{i+1}$ be the indices of the remaining active hypotheses and continue on to stage~$i+1$.
\end{enumerate} 

In other words, the procedure samples all active data streams until at least one of the active null hypotheses can be accepted or rejected, indicated by the stopping rule~\eqref{cont-samp}. At that point, stepdown rejection/acceptance rules are used in steps \ref{rej-step}/\ref{acc-step} to reject/accept some active null hypotheses. After updating the list of active hypotheses, the process is repeated until no active hypotheses remain.

\begin{remark}\label{rem:D.def} We make a few remarks about the generic procedure's definition.
\begin{enumerate}[(A)]
\item\label{D.no.confl} The relationships \eqref{AB.mono}-\eqref{ab.def} insure that there will never be a conflict between the rejections in Step~\eqref{rej-step} and the acceptances in Step~\eqref{acc-step}.  

\item Ties in the order statistics $\wtilde{\Lambda}^{(j)}_{n}$ in Step~\ref{step:ord} can be broken arbitrarily (at random, say) without affecting any of the error control properties proved below in Theorems~\ref{thm:DP} and \ref{thm:Dk}.

\item \label{rem:no.stand} If common critical values can be used for all data streams, that is, if $A_w^{(j)}=A_w^{(j')}=A_w$ and $B_w^{(j)}=B_w^{(j')}=B_w$ for all $j, j', w\in[J]$, then the standardizing functions can be dispensed with, i.e., we can take $\varphi^{(j)}(x)=x$ giving $a_j=A_j$ and $b_j=B_j$ for all $j\in[J]$.

\item The critical values $A_w^{(j)}, B_w^{(j)}$ may also depend on the current sample size $n$ of the test statistic $\Lambda^{(j)}(n)$ being compared with them, with only notational changes in the definition of the generic procedure and the properties proved below; to avoid overly cumbersome notation we have omitted this from the presentation. Standard group sequential stopping boundaries -- such as Pocock, O'Brien-Fleming, power family, and any others \citep[see][Chapters~2 and 4]{Jennison00} -- can be utilized for the individual test statistics in this way.

\item\label{D.gen.well.def} Note that the stopping time $n_i$ of the $i$th stage, given by \eqref{cont-samp}, is determined by the numbers~$c_i$ and $r_i$ of null hypotheses that have been rejected and accepted, respectively, during prior stages $1,\ldots,i-1$. Therefore this stopping rule is completely determined before the start of the $i$th stage and, in particular, unambiguously defined.

\end{enumerate}
\end{remark}

\textbf{Example.}\label{new.ex.sd} Before getting to specific examples of the generic stepdown procedure, we give a simplistic example to show the mechanics of the procedure. The following is a summary of an example appearing in \citet[][page~104]{Bartroff14b}, and details of the test statistics and critical values are given there and omitted here. All that is important for our purposes here is that there are $J=3$ chosen data streams, null/alternative hypothesis pairs $(H^{(j)}, G^{(j)})$, and sequential test statistics $\Lambda^{(j)}(n)$ with common critical values $A_w^{(j)}=A_w^{(j')}=A_w$ and $B_w^{(j)}=B_w^{(j')}=B_w$ for all $j, j', w\in\{1,2,3\}$, which are  given in the header of Table~\ref{tab:Bern-paths}. In particular, per Remark~\ref{rem:no.stand} above we take $a_j=A_j$, $b_j=B_j$, and $\wtilde{\Lambda}^{(j)}(n)=\Lambda^{(j)}(n)$ in the definition of the procedure. Table~\ref{tab:Bern-paths} contains three simulated sample paths and the critical values are given in the table's header. Let us focus on how the critical values determine the procedure's decisions to stop or continue sampling.  The values of the stopped test statistics are given in bold in the table.   
 
 On sample path~1, sampling proceeds until time $n_1=7$ when $H^{(1)}$ and $H^{(2)}$ are rejected because this is the first time any of the 3 test statistics exceed $B_1$ or fall below $A_1$. In particular, $H^{(1)}$ is rejected because $\Lambda^{(1)}(7)=2.03\ge B_1=1.93$ and $H^{(2)}$ is also rejected at this time because $\Lambda^{(2)}(7)=2.03\ge B_2=1.53$ and one null hypothesis (i.e., $H^{(1)}$) has already been rejected; the fact that $\Lambda^{(2)}(7)$ also exceeds $B_1$ was not necessary for rejecting $H^{(2)}$. Next, sampling of stream 3 is continued until time $n_2=10$ when $H^{(3)}$ is accepted because its test statistic falls below $A_1=-2.43$. Similarly, on sample path~2, after rejecting $H^{(1)}$ at time $n_1=7$, $H^{(2)}$ is then rejected at time $n_2=8$ because $\Lambda^{(2)}(8)$ exceeds $B_2=1.53$ and one null hypothesis (i.e., $H^{(1)}$) has already been rejected. $H^{(3)}$ is also accepted at time $n_2=8$ for the same reason as above. On sample path~3, all three null hypotheses are rejected at time $n_1=7$ because $\Lambda^{(1)}(7)=2.03\ge B_1$, $\Lambda^{(2)}(7)=2.03\ge B_2$ and one null hypothesis (i.e., $H^{(1)}$) has already been rejected, and $\Lambda^{(3)}(7)=1.22\ge B_3$ and two null hypotheses (i.e., $H^{(1)}$ and $H^{(2)}$) have already been rejected.

\begin{table}[h]
\caption{Three sample paths of a stepdown procedure for $J=3$ hypotheses using critical values $A_1=-2.34$, $A_2=-1.94$, $A_3=-1.27$, $B_1 = 1.93$, $B_2=1.53$, $B_3=.86$. The values of the stopped sequential statistics are in bold. From \citet[][page~104]{Bartroff14b}.}
\begin{center}
\begin{tabular}{cc|cccccccccc}
Data\\
Stream&&$n=1$&2&3&4&5&6&7&8&9&10\\\hline\hline
\multicolumn{12}{c}{\textit{Sample Path 1}}\\
\multirow{2}{*}{1}&$X_n^{(1)}$&0 &1 &1 &1 &1 &1 &1 &\\
&$\Lambda^{(1)}(n)$&-.41 &.00 &.41 &.81 &1.22 &1.62 &\textbf{2.03} &\\
\multirow{2}{*}{2}&$X_n^{(2)}$&1 &0 &1 &1 &1 &1 &1 &\\
&$\Lambda^{(2)}(n)$&.41 &.00 &.41 &.81 &1.22 &1.62 &\textbf{2.03} &\\
\multirow{2}{*}{3}&$X_n^{(3)}$&0 &1 &0 &0 &1 &0 &0 &0 &0 &0\\
&$\Lambda^{(3)}(n)$&-.41 &.00 &-.41 &-.81 &-.41 &-.81 &-1.22 &-1.62 &-2.03 &\textbf{-2.43}\\\hline
\multicolumn{12}{c}{\textit{Sample Path 2}}\\
\multirow{2}{*}{1}&&0 &1 &1 &1 &1 &1 &1 &\\
&&-.41 &.00 &.41 &.81 &1.22 &1.62 &\textbf{2.03} &\\
\multirow{2}{*}{2}&&1 &0 &0 &1 &1 &1 &1 &1 &\\
&&.41 &.00 &-.41 &.00 &.41 &.81 &1.22 &\textbf{1.62} &\\
\multirow{2}{*}{3}&&0 &1 &0 &0 &0 &0 &0 &0 &\\
&&-.41 &.00 &-.41 &-.81 &-1.22 &-1.62 &-2.03 &\textbf{-2.43} &
\\\hline
\multicolumn{12}{c}{\textit{Sample Path 3}}\\
\multirow{2}{*}{1}&&1 &0 &1 &1 &1 &1 &1 &\\
&&.41 &.00 &.41 &.81 &1.22 &1.62 &\textbf{2.03} &\\
\multirow{2}{*}{2}&&1 &1 &1 &0 &1 &1 &1 &\\
&&.41 &.81 &1.22 &.81 &1.22 &1.62 &\textbf{2.03} &\\
\multirow{2}{*}{3}&&0 &1 &0 &1 &1 &1 &1 &\\
&&-.41 &.00 &-.41 &.00 &.41 &.81 &\textbf{1.22} &\\\hline
\end{tabular}
\end{center}
\label{tab:Bern-paths}
\end{table}%

\subsubsection{A Stepdown Procedure Controlling $\gfdp$ and $\gfnp$}\label{sec:DP.procs}
The following step values\footnote{See Remark~\ref{rem:val.diff}.} were proposed by  \citet{Romano06b}. For $v\in[J]$ and $\gamma\in[0,1)$ define
\begin{align}
\overline{j}(t,v,\gamma)&=\min\{J,J+t-v,\lceil t/\gamma\rceil-1\}\qmq{for}t\in[\lfloor \gamma J\rfloor +1],\label{DP.jstar}\\
\overline{t}(v,\gamma)&=\min\left\{\lfloor \gamma J\rfloor+1, v, \left\lfloor \frac{\gamma (J-v)}{1-\gamma}\right\rfloor+1 \right\},\label{over.t}
\end{align}
omitting the third term in  the minimum in \eqref{DP.jstar} if $\gamma=0$.
Given a nondecreasing sequence $0\le \delta_1\le\ldots\le \delta_J\le 1$, for $v\in[J]$ and $\gamma\in[0,1)$ define
\begin{align*}
\eps(t,v,\gamma,\{\delta_j\})&=\delta_{\overline{j}(t,v,\gamma)}\qmq{for}t\in[\lfloor \gamma J\rfloor +1],\\
S_1(v,\gamma,\{\delta_j\})&= v\sum_{t=1}^{\overline{t}(v,\gamma)} \frac{\eps_t-\eps_{t-1}}{t}\qmq{where}\eps_t=\eps(t,v,\gamma,\{\delta_j\})\qmq{and}\eps_0=0,\\
D_1(\gamma,\{\delta_j\})&=\max_{0\le v\le J} S_1(v,\gamma,\{\delta_j\}).
\end{align*}
These quantities also depend on the total number~$J$ of null hypotheses but we have suppressed this in the notation since $J$ is fixed throughout. 

\begin{theorem}\label{thm:DP} Fix $\alpha, \beta\in (0,1)$ and $\gamma_1,\gamma_2\in [0,1)$.  Given any sequences of constants $0\le \delta_1\le\ldots\le \delta_J\le 1$ and $0\le \eta_1\le\ldots\le \eta_J\le 1$, define 
\begin{equation}\label{DP.gen.cons}
\alpha_j=\frac{\alpha\delta_j}{D_1(\gamma_1,\{\delta_{j'}\})},\quad \beta_j=\frac{\beta\eta_j}{D_1(\gamma_2,\{\eta_{j'}\})},\quad j\in[J]. 
\end{equation} If the test statistics and critical values satisfy the assumptions in Section~\ref{sec:stats} for these $\{\alpha_j,\beta_j\}_{j\in[J]}$, then the sequential stepdown procedure with step values \eqref{DP.gen.cons} satisfies 
\begin{equation*}
\gfdp(\theta)\le\alpha\qmq{and} \gfnp(\theta)\le \beta\qmq{for all}\theta\in\Theta 
\end{equation*}
regardless of the dependence between data streams.
\end{theorem}

\begin{remark}\label{rem:dn.DP} A special case of the theorem that will likely be useful in practice is given by 
\begin{equation}\label{dn.Holm.DP}
\delta_j=\frac{\lfloor \gamma_1 j\rfloor+1}{J+\lfloor \gamma_1 j\rfloor+1-j},\quad \eta_j=\frac{\lfloor \gamma_2 j\rfloor+1}{J+\lfloor \gamma_2 j\rfloor+1-j},\quad j\in[J].
\end{equation}
 Of course there are other possibilities, such as $\delta_j=\eta_j=j/J$, which give step values proportional to the ones used in the original FDR-controlling procedure of \citet{Benjamini95}, although \citet[][p.~44]{Romano06b} found these to be smaller (and thus less desirable) than the step values~\eqref{DP.gen.cons} given by \eqref{dn.Holm.DP}, for the most part.
 \end{remark}

\begin{remark}\label{rem:val.diff} The third term in \eqref{over.t} is a slight improvement over the corresponding third term in \citet[][Equation~(3.11)]{Romano06b}, and our proof holds in their fixed sample size setting, giving a slightly improved upper bound for the number of true hypotheses.
\end{remark}

\subsubsection{A Stepdown Procedure Controlling the $\fweI$ and $\fweII$}\label{Dk.procs}
The stepdown procedure in the following theorem utilizes step values proposed by \citet{Lehmann05b}. 

\begin{theorem}\label{thm:Dk} Fix $\alpha, \beta\in (0,1)$, $k_1,k_2\in[J]$, and define
\begin{equation}\label{cons.Dk}
\alpha_j=\frac{k_1\alpha}{J-(j-k_1)^+},\quad \beta_j=\frac{k_2\beta}{J-(j-k_2)^+},\quad j\in[J],
\end{equation} where $x^+=\max\{x,0\}$. If the test statistics and critical values satisfy the assumptions in Section~\ref{sec:stats} for these $\{\alpha_j,\beta_j\}_{j\in[J]}$, then the sequential stepdown procedure with step values \eqref{cons.Dk} satisfies 
\begin{equation}\label{FWE<.Dk}
\fweI(\theta) \le\alpha\qmq{and}\fweII(\theta) \le\beta\qmq{for all}\theta\in\Theta 
\end{equation}
regardless of the dependence between data streams.

\end{theorem}

\begin{remark}\label{rem:opt.Dk}
\citet[][Theorem~2.3]{Lehmann05b} exhibit a distribution of fixed sample size $p$-values for which the achieved (type~I) FWER is exactly the prescribed value~$\alpha$. The fixed sample size setting being a special case of the sequential setup considered here (by taking $X_1^{(j)}$ in \eqref{streams} to be the fixed sample size data and $X_n^{(j)}=\emptyset$ for $n>1$), applying their example to both true and false null hypotheses shows that there is a distribution for the data such that both the inequalities in \eqref{FWE<.Dk} are equalities. In this sense the bounds \eqref{FWE<.Dk} are sharp.
\end{remark}

\subsection{Stepup Procedures}\label{sec:U}

In this section we develop stepup procedures analogously to what was done for stepdown procedures in Section~\ref{sec:D}.

\subsubsection{The Generic Sequential Stepup Procedure}\label{sec:gen.U}

Here we define a generic sequential stepup procedure, special cases of which will be used to define the type~I and II $k$-FWER and $\gamma$-FDP controlling sequential procedures below.  We assume that  step values $\{\alpha_j, \beta_j\}_{j\in[J]}$ satisfying \eqref{UD.values} are given and that the test statistics and critical values satisfy the assumptions in Section~\ref{sec:stats} with respect to these values.

As for the generic sequential stepdown procedure in Section~\ref{sec:gen.D}, we describe the stepup procedure in terms of stages of sampling, between which reject/accept decisions are made, and we use the same notation $\mJ_i$, $n_i$, $r_i$, and $c_i$ as there, with $\mJ_1=[J]$, $n_0=0$, and $r_1=c_1 =0$. Then the $i$th stage of sampling ($i=1,2,\ldots$) of the \textbf{Generic Sequential Stepup Procedure} with step values $\{\alpha_j, \beta_j\}_{j\in[J]}$ proceeds as follows.

\begin{enumerate}
\item\label{fdrsample-step} Sample the active data streams $\{X_n^{(j)}\}_{j\in\mJ_i, \; n>n_{i-1}}$ until $n$ equals
\begin{equation}\label{fdrcont-samp}
n_i=\inf\left\{n>n_{i-1}: \wtilde{ \Lambda}^{(j(n,\ell))}(n)\not\in(a_{c_i+\ell}, b_{r_i+|\mJ_i|-\ell+1})\qmq{for some} \ell\in[|\mJ_i|]\right\},
\end{equation}
where $j(n,\ell)$ denotes the index of the $\ell$th ordered active standardized statistic at sample size $n$.

\item  
\begin{enumerate}
\item\label{fdrrej-step} If an upper boundary in \eqref{fdrcont-samp} was crossed, that is, if 
\begin{equation*}
\wtilde{ \Lambda}^{(j(n_i,\ell))}(n_i)\ge b_{r_i+|\mJ_i|-\ell+1}\qmq{for some} \ell\in[|\mJ_i|],
\end{equation*}
 then reject the $m_i\ge 1$ null hypotheses 
 \begin{equation*}
H^{(j(n_i,|\mJ_i|))}, H^{(j(n_i,|\mJ_i|-1))}, \ldots, H^{(j(n_i,|\mJ_i|-m_i+1))},
\end{equation*}
where 
\begin{equation}\label{fdrmjrej}
m_i=\max\left\{m\in [|\mJ_i|]: \wtilde{\Lambda}^{(j(n_i,|\mJ_i|-m+1))}(n_i)\ge b_{r_i+m} \right\},
\end{equation} 
and set $r_{i+1}=r_i+m_i$. Otherwise set $r_{i+1}=r_i$.

\item\label{fdracc-step} If a lower boundary in \eqref{fdrcont-samp} was crossed, that is, if 
\begin{equation*}
\wtilde{ \Lambda}^{(j(n_i,\ell))}(n_i)\le  a_{c_i+\ell}\qmq{for some} \ell\in[|\mJ_i|],
\end{equation*}  then accept the $m_i'\ge 1$ null hypotheses $$H^{(j(n_i,m_i'))}, H^{(j(n_i,m_i'-1))},\ldots, H^{(j(n_i,1))},$$ where 
\begin{equation*}
m_i'=\max\left\{m\in [|\mJ_i|] : \wtilde{\Lambda}^{(j(n_i,m))}(n_i)\le a_{c_i+m} \right\},
\end{equation*} and set $c_{i+1}=c_i+m_i'$. Otherwise set $c_{i+1}=c_i$.

\end{enumerate}

\item\label{fdrstop-step} Stop if there are no remaining active hypotheses, i.e., if $r_{i+1}+c_{i+1}=J$.  Otherwise, let $\mJ_{i+1}$ be the indices of the remaining active hypotheses and continue on to stage~$i+1$.
\end{enumerate} 

In other words, the procedure samples all active data streams until at least one of the active null hypotheses can be accepted or rejected, indicated by the stopping rule~\eqref{fdrcont-samp}. At that point, stepup rejection/acceptance rules are used in steps \ref{fdrrej-step}/\ref{fdracc-step} to reject/accept some active null hypotheses. After updating the list of active hypotheses, the process is repeated until no active hypotheses remain.

\begin{remark}
Points analogous to Remark~\ref{rem:D.def} apply to the generic sequential stepup procedure as well.
\end{remark}

\subsubsection{A Stepup Procedure Controlling $\gfdp$ and $\gfnp$}\label{sec:UP.procs}
The following step values were proposed by \citet{Romano06}. Given a nondecreasing sequence $0\le \delta_1\le\ldots\le \delta_J\le 1$, for $\gamma\in[0,1)$ and $v\in[J]$  define
\begin{align*}
S_2(v,\gamma, \{\delta_j\})&= v\delta_1+ v\sum_{v-J+1<s\le v,\; v\ge \lfloor \gamma(J-v+s)\rfloor+1} \frac{\delta_{J-v+s}-\delta_{J-v+s-1}}{s\vee (\lfloor \gamma (J-v+s)\rfloor +1)},\\
D_2(\gamma,\{\delta_j\})&=\max_{v\in [J]} S_2(v, \gamma, \{\delta_j\}).
\end{align*}
Here $x\vee y=\max\{x,y\}$. These quantities also depend on $J$ but we have suppressed this in the notation since $J$ is fixed throughout. 

\begin{theorem}\label{thm:UP} Fix $\alpha, \beta\in(0,1)$ and $\gamma_1,\gamma_2\in [0,1)$.  Given any sequences of constants $0\le \delta_1\le\ldots\le \delta_J\le 1$ and $0\le \eta_1\le\ldots\le \eta_J\le 1$, define
\begin{equation}\label{UP.gen.cons}
\alpha_j=\frac{\alpha\delta_j}{D_2(\gamma_1,\{\delta_{j'}\})},\quad \beta_j=\frac{\beta\eta_j}{D_2(\gamma_2,\{\eta_{j'}\})},\quad j\in[J]. 
\end{equation}
If the test statistics and critical values satisfy the assumptions in Section~\ref{sec:stats} for these $\{\alpha_j,\beta_j\}_{j\in[J]}$, then the sequential stepup procedure with step values \eqref{UP.gen.cons} satisfies 
\begin{equation*}
\gfdp(\theta)\le\alpha\qmq{and} \gfnp(\theta)\le\beta\qmq{for all}\theta\in\Theta 
\end{equation*}
regardless of the dependence between data streams.
\end{theorem}

\begin{remark}\label{rem:dn.UP} A special case of the theorem that will likely be useful in practice is given by \eqref{dn.Holm.DP}. Of course there are other possibilities, such as $\delta_j=\eta_j=j/J$, which give step values proportional to the ones used in the original FDR-controlling procedure of \citet{Benjamini95}, although \citet[][p.~1865]{Romano06} found these to be smaller (and thus less desirable), for the most part, than the step  values~\eqref{UP.gen.cons} given by \eqref{dn.Holm.DP}.
 \end{remark}

\begin{remark}\label{rem:opt.UP} \citet[][Theorem~4.1(ii)]{Romano06b} exhibit a joint distribution of $p$-values under which the procedure using step values~\eqref{cons.Uk} achieves $\gfdp(\theta) =\alpha$. Since, as mentioned in Remark~\ref{rem:opt.Dk}, the fixed-sample setting is a special case of the sequential setting, their example applies here as well, and the same argument gives a joint distribution under which $\gfnp=\beta$. Thus, their result provides a weak optimality property of the sequential stepup procedure.
 \end{remark}

\subsubsection{A Stepup Procedure Controlling $\fweI$ and $\fweII$}
The following step values were proposed by  \citet{Romano06}. Given a nondecreasing sequence $0\le \delta_1\le\ldots\le \delta_J\le 1$, for $k,v\in[J]$  define
\begin{align*}
S_3(v,k, \{\delta_j\})&= \frac{v\delta_{J-v+k}}{k}+ v\sum_{k<s\le v} \frac{\delta_{J-v+s}-\delta_{J-v+s-1}}{s},\\
D_3(k,\{\delta_j\})&=\max_{k\le v\le J} S_3(v, k, \{\delta_j\}).
\end{align*}
These quantities also depend on $J$ but we have suppressed this in the notation since $J$ is fixed throughout. 

\begin{theorem}\label{thm:Uk} Fix $\alpha, \beta\in (0,1)$ and  $k_1,k_2\in[J]$.  Given any sequences of constants $0\le \delta_1\le\ldots\le \delta_J\le 1$ and $0\le \eta_1\le\ldots\le \eta_J\le 1$, define
\begin{equation}\label{cons.Uk}
\alpha_j=\frac{\alpha\delta_j}{D_3(k_1,\{\delta_{j'}\})},\quad \beta_j=\frac{\beta\eta_j}{D_3(k_2,\{\eta_{j'}\})},\quad j\in[J]. 
\end{equation}
If the test statistics and critical values satisfy the assumptions in Section~\ref{sec:stats} for these $\{\alpha_j,\beta_j\}_{j\in[J]}$, then the sequential stepup procedure with step values \eqref{cons.Uk} satisfies
\begin{equation*}
\fweI(\theta) \le\alpha\qmq{and} \fweII(\theta) \le\beta\qmq{for all}\theta\in\Theta 
\end{equation*}
regardless of the dependence between data streams.
\end{theorem}

\begin{remark}\label{rem:dn.Uk} A special case of the theorem that will likely be useful in practice is given by the constants
\begin{equation}\label{dn.Uk}
\delta_j=\frac{k_1}{J-(j-k_1)^+},\quad \eta_j=\frac{k_2}{J-(j-k_2)^+},\quad j\in[J],
\end{equation} which are proportional to those proposed by \citet{Hommel88} and \citet{Lehmann05b}, as well as \eqref{cons.Dk} in the proposed stepdown procedure. Other possibilities exist, such as  $\delta_j=\eta_j=j/J$, but \citet[][p.~1859]{Romano06} computed the resulting step  values~\eqref{cons.Uk} for both of these choices and found that those given by \eqref{dn.Uk} to be larger (and hence more desirable) than those given by $j/J$ for large or small values of $j$, and smaller for moderate values of $j$ but differing by relatively little in this case.
 \end{remark}

\begin{remark}\label{rem:opt.Uk} \citet[][Theorem~3.1(ii)]{Romano06} exhibited a joint distribution of $p$-values under which the procedure using step values~\eqref{cons.Uk} achieves $\fweI(\theta)=\alpha$. Since, as mentioned in Remark~\ref{rem:opt.Dk}, the fixed-sample setting is a special case of the sequential setting, their example applies here as well, and the same argument gives a joint distribution under which $\fweII(\theta)=\beta$. Thus, their result provides a weak optimality property of the sequential stepup procedure.
 \end{remark}
 
\section{Versions of the Procedures  Controlling only the Type I Generalized Error Rate}\label{sec:rej}

In this section we describe versions of the above procedures which only stop early to reject (rather than accept) null hypotheses  and thus which only explicitly control the corresponding type~I generalized error rate, recorded in Theorems~\ref{thm:D.rej} and \ref{thm:U.rej}.  For this reason we refer to them as ``rejective'' versions of the procedures.  The rejective procedures may be preferable to the statistician in certain situations such as when (a) a null hypothesis being true represents the system being ``in control'' and therefore continued sampling (rather than stopping) is desirable, (b) there is a maximum sample size imposed on the data streams preventing achievement of the error bounds \eqref{typeI}-\eqref{typeII}, or (c) the type~II generalized error rate~$\beta$ is not well-motivated.  In any of theses cases, the statistician may prefer to drop the requirement that the type~II generalized error rate be strictly controlled at $\beta$ and use one of the rejective procedures which, roughly speaking, are similar to those above but ignore the lower stopping boundaries $A_w^{(j)}$. However, even if $\beta$ is not well motivated but the statistician prefers early stopping under the null hypotheses, then we encourage the use of the procedures above while treating $\beta$ as  a parameter to be chosen to give a procedure with other desirable operating characteristics, such as expected total or streamwise maximum sample size.

The setup for rejective procedures is similar to that above with a few modifications. Let the data streams $X_n^{(j)}$, test statistics $\Lambda^{(j)}(n)$, and parameters $\theta^{(j)}$ and $\theta$ be as in Section~\ref{sec:setup}. Since only the type~I error rate, i.e., $\gfdp$ or $\fweI$, will be explicitly controlled we only require specification of null hypotheses $H^{(j)}\subseteq\Theta^{(j)}$ and not alternative hypotheses $G^{(j)}$. Accordingly we modify the definition of the false hypotheses~\eqref{F} to be 
\begin{equation*}
\mF(\theta)=\{j\in[J]:\theta^{(j)}\not\in H^{(j)}\},
\end{equation*}
and the true hypotheses $\mT(\theta)$ are still given by \eqref{T}. As mentioned in point~(b) above, one situation in which the rejective procedures may be be desirable is when there is a streamwise maximum sample size (or ``truncation point'') $\overline{N}$ which we now assume in this section, although with only notational changes what follows could be formulated without a truncation point or with sample sizes other than  $1,\ldots,\overline{N}$. 

Given a sequence of step values $0\le\alpha_1\le\ldots\le\alpha_J\le 1$, we assume that the test statistics~$\Lambda^{(j)}(n)$ have associated critical values $B_1^{(j)},\ldots,B_J^{(j)}$ satisfying 
\begin{equation}
\label{typeI.rej}
P_{\theta^{(j)}}\left( \Lambda^{(j)}(n)\ge B_w^{(j)}\;\mbox{for some $n\le\overline{N}$}\right)\le \alpha_w\qmq{for all} \theta^{(j)} \in H^{(j)},
\end{equation} for each $w\in[J]$, as well as \eqref{AB.mono} and \eqref{B.mono} without loss of generality. We let the standardizing functions~$\varphi^{(j)}$ be any increasing functions such that $b_w=\varphi^{(j)}(B_w^{(j)})$ does not depend on $j$, and $\wtilde{\Lambda}^{(j)}(n)=\varphi^{(j)}(\Lambda^{(j)}(n))$ denote the standardized statistics.

In the next two sections we give the rejective versions of the generic stepdown and stepup procedures in Sections~\ref{sec:gen.D} and \ref{sec:gen.U}, respectively, and state their type~I generalized error control properties in Theorems~\ref{thm:D.rej} and \ref{thm:U.rej}. The proofs are similar to the proofs of the corresponding theorems in Section~\ref{sec:I&II} and are thus omitted.

\subsection{Rejective Sequential Stepdown Procedures}\label{sec:RD.gen}
Letting $x\wedge y=\min\{x,y\}$ and with the same notation as in Section~\ref{sec:gen.D}, the $i$th stage ($i=1,2,\ldots$) of the \textbf{Generic Rejective Sequential Stepdown Procedure} with step values $\{\alpha_j\}_{j\in[J]}$ proceeds as follows.

\begin{enumerate}
\item Sample the active streams $\{X_n^{(j)}\}_{j\in\mJ_i, \; n>n_{i-1}}$ until $n$ equals
\begin{equation}\label{cont-samp.rej}n_i=\oN\wedge \inf\left\{n>n_{i-1}: \wtilde{\Lambda}^{(j)}(n) \ge b_{r_i+1}\qmq{for some} j\in\mJ_i \right\}.\end{equation}

\item If $n_i=\oN$ and no test statistic has crossed the critical value in \eqref{cont-samp.rej}, accept all active null hypotheses and terminate the procedure.  Otherwise, proceed to Step~\ref{step:ord.rej}.

\item\label{step:ord.rej} Order the active test statistics
\begin{equation*}
\wtilde{\Lambda}^{(j(n_i,1))}(n_i)\le \wtilde{\Lambda}^{(j(n_i,2))}(n_i)\le \ldots\le \wtilde{\Lambda}^{(j(n_i,|\mJ_i|))}(n_i)
\end{equation*} 
and reject the $m_i\ge 1$ null hypotheses 
\begin{equation*}
H^{(j(n_i,|\mJ_i|))}, H^{(j(n_i,|\mJ_i|-1))}, \ldots, H^{(j(n_i,|\mJ_i|-m_i+1))},
\end{equation*}
where 
\begin{equation*}
m_i=\max\left\{m\in [|\mJ_i|]: \wtilde{\Lambda}^{(j(n_i,\ell))}(n_i)\ge b_{r_i+|\mJ_i|-\ell+1}\qmq{for all}\ell=|\mJ_i|-m+1,\ldots, |\mJ_i| \right\}.
\end{equation*} 

\item If $r_i+m_i=J$ or $n_i=\oN$, terminate the procedure. Otherwise, set $r_{i+1}=r_i+m_i$, let $\mJ_{i+1}$ be the indices of the remaining hypotheses, and continue on to stage $i+1$.
\end{enumerate} 

\begin{remark}
Points analogous to Remark~\ref{rem:D.def} (except Point~\ref{D.no.confl} which doesn't apply since there is no early acceptance rule) apply to the generic rejective sequential stepdown procedure as well.
\end{remark}

\begin{theorem}\label{thm:D.rej} Fix $\alpha \in (0,1)$.  
\begin{enumerate}
\item\label{part:rej.DP}  Fix $\gamma_1\in[0,1)$. Given any sequence of constants $0\le \delta_1\le\ldots\le \delta_J\le 1$ let $\alpha_j$ be given by \eqref{DP.gen.cons}. If the test statistics and critical values satisfy the assumptions above for these $\alpha_j$, then the rejective sequential stepdown procedure with step values \eqref{DP.gen.cons} satisfies $\gfdp(\theta)\le\alpha$ regardless of the dependence between data streams.

\item\label{part:rej.Dk} Fix $k_1 \in[J]$ and let $\alpha_j$ be given by \eqref{cons.Dk}. If the test statistics and critical values satisfy the assumptions above for these $\alpha_j$, then the rejective sequential stepdown procedure with step values \eqref{cons.Dk} satisfies $\fweI(\theta) \le\alpha$ regardless of the dependence between data streams.
\end{enumerate}
\end{theorem}

\begin{remark}
As mentioned in Remark~\ref{rem:dn.DP}, the $\delta_j$ given in \eqref{dn.Holm.DP} may be useful in practice for the procedure in Part~\ref{part:rej.DP} of the theorem. Also, the weak optimality mentioned in Remark~\ref{rem:opt.Dk} applies as well to the rejective procedure in Part~\ref{part:rej.Dk} of the theorem.
\end{remark}

\subsection{Rejective Sequential Stepup Procedures}\label{sec:RU.gen}

With the same notation as in Section~\ref{sec:gen.U}, the $i$th stage ($i=1,2,\ldots$) of the \textbf{Generic Rejective Sequential Stepup Procedure} with step values $\{\alpha_j\}_{j\in[J]}$ proceeds as follows.

 \begin{enumerate}
\item Sample the active data streams $\{X_n^{(j)}\}_{j\in\mJ_i, \; n>n_{i-1}}$ until $n$ equals
\begin{equation}\label{fdrcont-samp.rej}
n_i=\oN\wedge \inf\left\{n>n_{i-1}: \wtilde{ \Lambda}^{(j(n,\ell))}(n)\ge b_{r_i+|\mJ_i|-\ell+1}\qmq{for some} \ell\in[|\mJ_i|]\right\}.
\end{equation}

\item If $n_i=\oN$ and no test statistic has crossed its corresponding critical value in \eqref{fdrcont-samp.rej}, accept all active null hypotheses and terminate the procedure.  Otherwise, proceed to Step~\ref{step:rej.U.rej}.

\item\label{step:rej.U.rej} Reject the $m_i\ge 1$ null hypotheses 
 \begin{equation*}
H^{(j(n_i,|\mJ_i|-m_i+1))}, H^{(j(n_i,|\mJ_i|-m_i+2))}, \ldots H^{(j(n_i,|\mJ_i|))}, \end{equation*}
where 
\begin{equation*}
m_i=\max\left\{m\in [|\mJ_i|]: \wtilde{\Lambda}^{(j(n_i,|\mJ_i|-m+1))}(n_i)\ge b_{r_i+m} \right\}.
\end{equation*}

\item  If $r_i+m_i=J$ or $n_i=\oN$, terminate the procedure. Otherwise, set $r_{i+1}=r_i+m_i$, let $\mJ_{i+1}$ be the indices of the remaining hypotheses, and continue on to stage $i+1$.
\end{enumerate} 

\begin{remark}
Points analogous to Remark~\ref{rem:D.def} (except Point~\ref{D.no.confl} which doesn't apply since there is no early acceptance rule) apply to the generic rejective sequential stepup procedure as well.
\end{remark}

\begin{theorem}\label{thm:U.rej} Fix $\alpha \in (0,1)$.  
\begin{enumerate}
\item\label{part:rej.UP}  Fix $\gamma_1\in[0,1)$. Given any sequence of constants $0\le \delta_1\le\ldots\le \delta_J\le 1$ let $\alpha_j$ be given by \eqref{UP.gen.cons}. If the test statistics and critical values satisfy the assumptions above for these $\alpha_j$, then the rejective sequential stepup procedure with step values \eqref{UP.gen.cons} satisfies $\gfdp(\theta)\le\alpha$ regardless of the dependence between data streams.

\item\label{part:rej.Uk} Fix $k_1 \in[J]$. Given any sequence of constants $0\le \delta_1\le\ldots\le \delta_J\le 1$ let $\alpha_j$ be given by \eqref{cons.Uk}. If the test statistics and critical values satisfy the assumptions above for these $\alpha_j$, then the rejective sequential stepup procedure with step values \eqref{cons.Uk} satisfies $\fweI(\theta) \le\alpha$ regardless of the dependence between data streams.
\end{enumerate}
\end{theorem}

\begin{remark}
As mentioned in Remarks~\ref{rem:dn.UP} and \ref{rem:dn.Uk}, the $\delta_j$ given in \eqref{dn.Holm.DP} and \eqref{dn.Uk} may be useful in practice for the procedures in Parts~\ref{part:rej.UP} and \ref{part:rej.Uk} of the theorem, respectively. Also, the weak optimality mentioned in Remarks~\ref{rem:opt.UP} and \ref{rem:opt.Uk} applies as well to the rejective procedures in Parts~\ref{part:rej.UP} and \ref{part:rej.Uk} of the theorem, respectively.
\end{remark}

\section{Implementation}\label{sec:imp}

\subsection{Simple vs.\ Simple Hypotheses}\label{sec:simple}
In this section we briefly discuss constructing individual test statistics and critical values satisfying \eqref{typeI}-\eqref{typeII} (or \eqref{typeI.rej} for the rejective versions of the procedures). More  complete discussions, including discussion of testing more general composite hypotheses and examples, are given in  \citet{Bartroff14c,Bartroff14b}. Here we focus on simple hypotheses and those  that can be approximated by simple hypotheses and in Theorem~\ref{thm:simple} we give closed-form expressions for the critical values $A_w^{(j)}, B_w^{(j)}$ satisfying \eqref{typeI}-\eqref{typeII} to a very close approximation, and which are based on the closed-form, widely-used Wald approximations for the sequential probability ratio test (SPRT). Sequential test statistics and critical values for other testing situations, including composite hypotheses and nuisance parameter problems, are covered in the two papers above and more generally in the texts \citet{Bartroff13} and \citet{Siegmund85}.

Focusing on a stream $j$ for which $H^{(j)}$ and $G^{(j)}$ are simple hypotheses, a natural choice for the test statistic~$\Lambda^{(j)}(n)$ is the log-likelihood ratio because of its strong optimality property of the  resulting (single hypothesis) test, the SPRT; see \citet{Chernoff72}. In order to express the likelihood ratio test in  a simple form, we now make the additional assumption that each data stream $X_1^{(j)},X_2^{(j)},\ldots$ constitutes independent and identically distributed data. However, we stress that this independence assumption is limited to \emph{within} each stream so that, for example, elements of $X_1^{(j)},X_2^{(j)},\ldots$ may be correlated with (or even identical to) elements of another stream $X_1^{(j')},X_2^{(j')},\ldots$.   Formally we represent the simple null and alternative hypotheses $H^{(j)}$ and $G^{(j)}$ by the corresponding distinct density functions $h^{(j)}$ (null) and $g^{(j)}$ (alternative) with respect to some common $\sigma$-finite measure $\mu^{(j)}$.  The parameter space~$\Theta^{(j)}$ corresponding to this data stream is the set of all densities $f$ with respect to $\mu^{(j)}$, and $H^{(j)}$ is considered true if the actual density~$f^{(j)}$ satisfies $f^{(j)}=h^{(j)}$ $\mu^{(j)}$-a.s., and is false if $f^{(j)}=g^{(j)}$ $\mu^{(j)}$-a.s. The SPRT for testing $H^{(j)}: f^{(j)}=h^{(j)}$ vs.\ $G^{(j)}: f^{(j)}=g^{(j)}$ with type I and II error probabilities $\alpha$ and $\beta$, respectively, utilizes the simple log-likelihood ratio test statistic 
\begin{equation}\label{simpleLLR}
\Lambda^{(j)}(n)=\sum_{i=1}^n \log\left(\frac{g^{(j)}(X_{i}^{(j)})}{h^{(j)}(X_{i}^{(j)})}\right)
\end{equation} and samples sequentially until $\Lambda^{(j)}(n)\not\in(A, B)$, where the critical values $A, B$ satisfy
\begin{align}
P_{h^{(j)}}(\Lambda^{(j)}(n)\ge B\;\mbox{some $n$,}\; \Lambda^{(j)}(n')>A\;\mbox{all $n'<n$})&\le \alpha\label{SPRT-typeI}\\
P_{g^{(j)}}(\Lambda^{(j)}(n)\le A\;\mbox{some $n$,}\; \Lambda^{(j)}(n')<B\;\mbox{all $n'<n$})&\le\beta.\label{SPRT-typeII}
\end{align} The most simple and widely-used method for finding $A$ and $B$ is to use the closed-form \emph{Wald-approximations} $A=A_W(\alpha,\beta)$ and $B=B_W(\alpha,\beta)$, where 
\begin{equation}\label{myAB}
A_W(a,b) =\log\left(\frac{b}{1-a}\right)+\rho,\quad B_W(a,b)=\log\left(\frac{1-b}{a}\right)-\rho
\end{equation} for $a,b\in(0,1)$ such that $a+b\le 1$ and a fixed quantity $\rho\ge 0$. See \citet[][Section~3.3.1]{Hoel71} for a derivation of the $\rho=0$ case and, based on Brownian motion approximations, \citet[][p.~50 and Chapter~X]{Siegmund85} derives the value $\rho=.583$ which has been used to improve the approximation for continuous random variables. Although, in general, the inequalities in \eqref{SPRT-typeI}-\eqref{SPRT-typeII} only hold approximately  when  using the Wald approximations $A=A_W(\alpha,\beta)$ and $B=B_W(\alpha,\beta)$, \citet{Hoel71} show that the actual type I and II error probabilities can only exceed $\alpha$ or $\beta$ by a small amount in the worst case, and the difference approaches $0$ for small $\alpha$ and $\beta$, which is relevant in the present multiple testing situation where we will utilize small values of $\alpha$ and $\beta$, i.e., certain fractions of the actual prescribed error rates. 

Next we use the Wald approximations to construct closed-form critical values $A_w^{(j)}$, $B_w^{(j)}$ satisfying \eqref{typeI}-\eqref{typeII} up to Wald's approximation. Specifically, given step  values $\{\alpha_j,\beta_j\}$, we show that when using \eqref{fdrAsBs}, the left-hand-sides of \eqref{typeI}-\eqref{typeII} equal the same quantities one would get using Wald's approximations with $\alpha_j,\beta_j$ in place of $\alpha, \beta$. This generalizes results of \citet{Bartroff14c,Bartroff14b} which gave Wald approximations for the specific step values $\{\alpha_j,\beta_j\}$ proposed for the FWER- and FDR-controlling procedures, respectively, given there.

\begin{theorem}\label{thm:simple} Fix $\{\alpha_j,\beta_j\}_{j\in[J]}$ satisfying \eqref{UD.values} and $\alpha_1+\beta_1 \le 1$, and $\rho\ge 0$. Suppose that, for a certain data stream~$j$, the associated hypotheses $H^{(j)}: f^{(j)}=h^{(j)}$ and $G^{(j)}: f^{(j)}=g^{(j)}$ are simple. For $a,b\in(0,1)$ such that $a+b\le 1$ let $\alpha_W^{(j)}(a,b)$ and $\beta_W^{(j)}(a,b)$ be the values of the probabilities on the left-hand sides of \eqref{SPRT-typeI} and \eqref{SPRT-typeII}, respectively, when $\Lambda^{(j)}(n)$ is given by \eqref{simpleLLR} and $A=A_W(a,b)$ and $B=B_W(a,b)$ are given by the Wald approximations \eqref{myAB}. For $w\in[J]$ let 
\begin{equation*}
\wtilde{\alpha}_w=\frac{\alpha_1 (1-\beta_w)}{1-\beta_1}\qmq{and}  \wtilde{\beta}_w=\frac{\beta_1 (1-\alpha_w)}{1-\alpha_1},
\end{equation*}
and let $p_w^{(j)}$ and $q_w^{(j)}$ denote the left-hand-sides of \eqref{typeI} and \eqref{typeII}, respectively, with $A_w^{(j)}$, $B_w^{(j)}$ given by
\begin{equation}\label{fdrAsBs}
A_w^{(j)}=\log\left(\frac{\beta_w(1-\beta_1)}{1-\beta_1-\alpha_1(1-\beta_w)}\right) +\rho,\quad B_w^{(j)}=\log\left(\frac{1-\alpha_1-\beta_1(1-\alpha_w)}{\alpha_w(1-\alpha_1)}\right)-\rho.
\end{equation}
Then, for all $w\in[J]$,
\begin{gather}
\alpha_w+\wtilde{\beta}_w \le 1,\quad \wtilde{\alpha}_w + \beta_w \le 1,\label{sum<=1}\\
p_w^{(j)}=\alpha_W^{(j)}(\alpha_w,\wtilde{\beta}_w),\qmq{and} q_w^{(j)}=\beta_W^{(j)}(\wtilde{\alpha}_w, \beta_w)\label{fdraH=aS}
\end{gather}
and therefore \eqref{typeI}-\eqref{typeII} hold, up to Wald's approximation, when using the critical values~\eqref{fdrAsBs}.
\end{theorem} 

We remark that the $\rho=0$ case of Theorem~\ref{thm:simple} holds without the independence assumption on $X_1^{(j)}, X_2^{(j)},\ldots$ made in this section, since this original form of Wald's approximations does not require this. 

\subsection{Group Sequential Testing}\label{sec:group.seq}

As mentioned above, the setup considered here is general enough to admit group sequential sampling as a special case and the popular methods for choosing group sequential stopping boundaries -- such as Pocock's \citeyearpar{Pocock77} test and O'Brien and Fleming's \citeyearpar{OBrien79} test, which we consider as examples here -- can be utilized.  See also \citet[][Chapters~2.4 and 2.5]{Jennison00} for these tests, whose setup  we follow.  Both Pocock's and O'Brien and Fleming's tests, in their original forms, utilize a fixed maximum number~$g$ of groups and only allow early stopping to reject the corresponding null hypothesis; if the null is not rejected at or before the $g$th group then it is accepted. This is precisely the form of the rejective procedures defined in Section~\ref{sec:rej}, which we now consider; the last paragraph in this section discusses group sequential tests that allow early rejection or acceptance of the null hypothesis. To utilize Pocock's test of the null hypothesis~$H^{(j)}: \theta^{(j)}=0$ about the average difference~$\theta^{(j)}$ in treatment effects with at most $g$ groups all of size~$m$ (although groups of unequal sizes can be handled with only minor notational burden), let $X_n^{(j)}=(D_{(n-1)m+1}^{(j)},D_{(n-1)m+2}^{(j)},\ldots,D_{nm}^{(j)})$, $n\in[g]$, be the vector of observed differences~$D_i^{(j)}$ in the $n$th group. Pocock's test statistic can be written 
\begin{equation}\label{L.Pocock}
\Lambda^{(j)}(n)=\abs{\frac{1}{\sqrt{nm\sigma^2}}\sum_{i=1}^{nm}D_i^{(j)}}\qmq{for}n\in[g],
\end{equation} where $\sigma^2$ is the known variance of the $D_i^{(j)}$. Given $\alpha\in(0,1)$, the $\alpha$-level version of the test stops after group $n\in[g]$ and rejects $H^{(j)}$ if $\Lambda^{(j)}(n)\ge C_P(\alpha)$, accepting $H^{(j)}$ if no rejection has occurred by the $g$th group.  Here $C_P(\alpha)$ is a constant (the subscript $P$ for Pocock) calculated to make the type~I error probability of this test no greater than $\alpha$, i.e.,
\begin{equation}\label{Pocock.typeI} 
P_{\theta^{(j)}=0}\left( \Lambda^{(j)}(n)\ge C_P(\alpha)\;\mbox{for some $n\in[g]$}\right)\le \alpha\qmq{for any}\alpha\in(0,1).
\end{equation} Calculation of $C_P(\alpha)$ is well-understood and included in many standard software packages; see \citet[][Chapter~19]{Jennison00}. 

To utilize the Pocock test as the $j$th component test in a rejective sequential stepup or stepdown procedure defined in Section~\ref{sec:rej}, let $\overline{N}=g$, $\Lambda^{(j)}(n)$ be as in \eqref{L.Pocock} for $n\in[g]$, and $B_w^{(j)}=C_P(\alpha_w)$ for $w\in[J]$ where $\alpha_w$ are the given step values. By these definitions and those of the rejective procedures we see that $H^{(j)}$ will be rejected at the first stage $n\in[\overline{N}]=[g]$ where $\Lambda^{(j)}(n)$ crosses a certain boundary~$B_w^{(j)}$, and accepted otherwise. All that remains to check that Theorems~\ref{thm:D.rej} and \ref{thm:U.rej} are in force is to verify that \eqref{typeI.rej} holds, whose left-hand side is equal to $$P_{\theta^{(j)}=0}\left( \Lambda^{(j)}(n)\ge C_P(\alpha_w)\;\mbox{for some $n\in[g]$}\right)$$ which, by \eqref{Pocock.typeI}, is no greater than $\alpha_w$.

O'Brien and Fleming's test can be applied similarly but with the slightly different test statistic 
\begin{equation}\label{L.OF}
\Lambda^{(j)}(n)=\abs{\frac{1}{\sqrt{gm\sigma^2}}\sum_{i=1}^{nm}D_i^{(j)}}\qmq{for}n\in[g],
\end{equation} which differs from \eqref{L.Pocock} by a factor of $\sqrt{g/n}$. This test stops to reject $H^{(j)}$ at the earliest stage~$n\in[g]$ such that $\Lambda^{(j)}(n)\ge C_{OF}(\alpha)$, constants satisfying
\begin{equation}\label{OF.typeI} 
P_{\theta^{(j)}=0}\left( \Lambda^{(j)}(n)\ge C_{OF}(\alpha)\;\mbox{for some $n\in[g]$}\right)\le \alpha\qmq{for any}\alpha\in(0,1).
\end{equation} 
Using this test as a component test in a rejective procedure is similar to that for Pocock's test but taking $B_w^{(j)}=C_{OF}(\alpha_w)$. As above, \eqref{OF.typeI} guarantees that \eqref{typeI.rej} holds, and hence Theorems~\ref{thm:D.rej} and \ref{thm:U.rej} are in force.

Neither Pocock's nor O'Brien and Fleming's tests stop early to accept the null hypothesis, but other popular group sequential tests do allow this behavior, such as power family tests \citep[see][Chapter~5]{Jennison00}. These tests can be used as component tests in the sequential stepup or stepdown procedures in Section~\ref{sec:I&II} in a similar way as the discussion above for rejective procedures with the minor notational burden of including a maximum sample size $\overline{N}$, equal to the maximum number of groups in this group sequential setting. Of course the choice of $\overline{N}$, as well as the group size (e.g., $m$ in the discussion above) may affect the ability to achieve the needed type~I and II error probabilities~\eqref{typeI} and \eqref{typeII}, but this issue is not unique to multiple testing considerations and must be considered in group sequential testing of a single null hypothesis as well.

\section{Numerical Comparisons}\label{sec:sim}
\subsection{Introduction and Setup}\label{sec:sim.setup}
Although a comprehensive comparison of the sequential stepup and stepdown procedures proposed above is beyond the scope of this article, in this section we give a comparison in the particular setting of inference about the means of strongly positively correlated Gaussian data streams; \citet{Muller07} note that this setting is still one of  the most widely used in applications involving multiple testing.

\begin{figure}[p]
\begin{center}
\scalebox{1}{\includegraphics{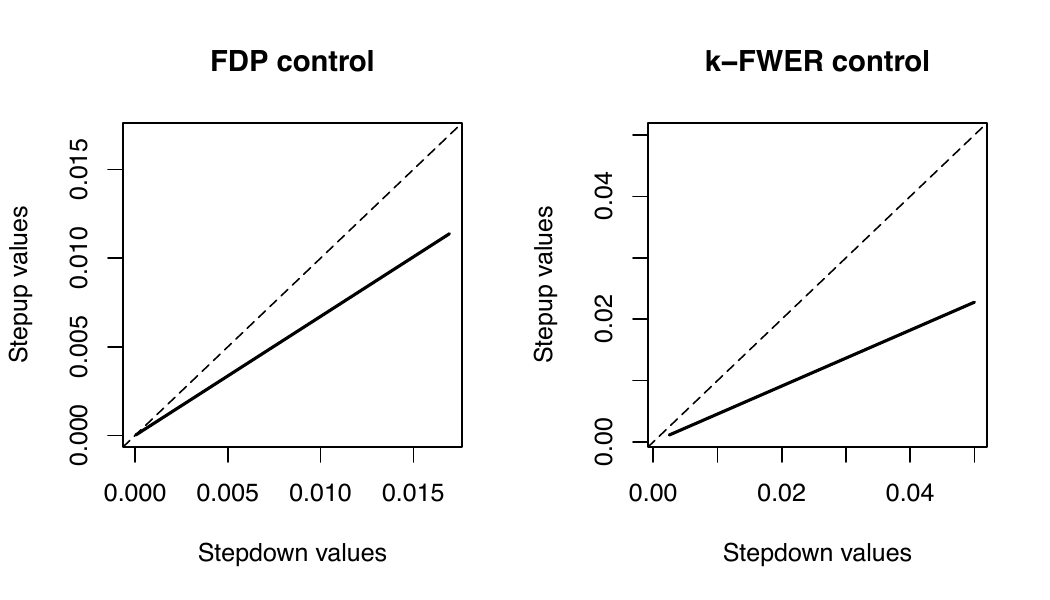}}
\caption{Stepdown versus stepup values (solid lines) for testing $J=500$ null hypotheses with $\alpha=.05$. In the left panel, the stepdown and stepup values $\alpha_j$ are given by \eqref{DP.gen.cons} and \eqref{UP.gen.cons}, respectively, both with $\delta_j$ given by \eqref{dn.Holm.DP} and  $\gamma_1=0.1$.  In the right panel, the stepdown and stepup values $\alpha_j$ are given by \eqref{cons.Dk} and \eqref{cons.Uk}, respectively, with $k_1=25$ and $\delta_j$ given by \eqref{dn.Uk} in the latter case. The identity line is dashed.}
\label{fig:step.vals}
\end{center}
\end{figure}

As mentioned in Section~\ref{sec:stats}, if a fixed sample stepup procedure  uses the same (or larger) step values $\{\alpha_j\}$ as a stepdown procedure, then the stepup procedure is preferred because  it will reject more null hypotheses and hence be more powerful while not exceeding the prescribed multiple testing error bound~$\alpha$.  The same statement holds about the rejective sequential procedures in  Section~\ref{sec:rej}, and an analogous statement holds about the sequential procedures in Section~\ref{sec:I&II} which control both type~I and II generalized error rates and their step values $\{\alpha_j, \beta_j\}$ in which case ``more powerful'' means less conservative type~I and II error control below the prescribed values~$\alpha$ and $\beta$.  However there is no such simple ``dominating'' relationship between the values of the stepup and stepdown procedures proposed above.  For example, Figure~\ref{fig:step.vals} contains plots of the stepdown versus stepup values $\alpha_j$ defined in Sections~\ref{sec:D} and \ref{sec:U}, respectively, for $\alpha=.05$, $J=500$ null hypotheses, and $\gamma_1=.1$ for FDP control (left panel) and $k_1=25$ for \fweI\, control (right panel). In both panels the solid line is below the dotted identity line indicating that each stepdown value exceeds its corresponding stepup value.

Thus, to investigate the efficiency and overall performance of the proposed sequential stepdown and stepup  procedures, simulation studies were performed to estimate their operating characteristics. For this, $J$ streams of Gaussian data were repeatedly simulated in order to consider a battery of tests of the form
\begin{equation}\label{z.test.hyp}
H^{(j)}: \theta^{(j)}\le 0\qmq{vs.} G^{(j)}: \theta^{(j)}\ge 1
\end{equation}
about the mean~$\theta^{(j)}$ of the $j$th data stream.  The proposed procedures, with their strict error control regardless of dependence, will probably be most useful in settings with strongly positively correlated data streams. \label{pos.refs}For example, about multiple testing problems which arise in genetic association studies by comparing many possible statistical models for genetic data, \citet[][p.~24]{Zheng12} remark that typically ``all genetic models under consideration are positively correlated.'' And in randomized multi-arm clinical trials, \citet[][p.~4369]{Freidlin08} note that ``individual comparisons are positively correlated due to the use of the same control arm.'' To create such a setting of strongly positively correlated data streams, the collection $(X_n^{(1)},\ldots, X_n^{(J)})$ of the $n$th observations
from the $J$ data streams were simulated as a $J$-dimensional multivariate normal distribution with mean~$\theta=(\theta^{(1)},\ldots,\theta^{(J)})$ and covariance matrix
\begin{equation}\label{covar}
\sigma^2\left[
\begin{array}{cccc}
 1 &.95&\cdots&.95   \\
 .95&1&\cdots&.95\\
 \vdots&\vdots&\ddots&\vdots\\
 .95&.95&\cdots&1
\end{array}
\right].
\end{equation} Constant correlation models such as this have recently been popular in the study of genetic correlation structure \citep{Lee11,Hardin13}, and for us \eqref{covar} provides a convenient way of generating a large number of data streams with strong positive correlation.  In the studies that follow we have chosen $\sigma=2$ to give tests of reasonable length. We note that, while the collection $(X_n^{(1)},\ldots, X_n^{(J)})$ of the $n$th observations  was generated using the above distribution, successive observations $(X_n^{(1)},\ldots, X_n^{(J)})$, $(X_{n+1}^{(1)},\ldots, X_{n+1}^{(J)})$ were generated independently. The test statistics \eqref{simpleLLR} were used with $\theta^{(j)}=0$ vs.\ $\theta^{(j)}=1$ as surrogate hypotheses, reducing to
$$\Lambda^{(j)}(n)=\frac{1}{\sigma^2}\left(\sum_{i=1}^n X_i^{(j)}-\frac{n}{2}\right)$$ in this case, and the critical values \eqref{fdrAsBs} were used with $\rho=.583$ and $\{\alpha_j,\beta_j\}$ as described below. The results of simulation studies in this setting are reported in Section~\ref{sec:P.table} for $\gfdp$ and $\gfnp$ control, and Section~\ref{sec:k.table} for $\fweI$ and $\fweII$ control. Finally, in Section~\ref{sec:t}, the assumption of known variance is dropped and Student's $t$-tests of composite hypotheses are considered.

\subsection{Study of Procedures Controlling $\gfdp$ and $\gfnp$}
\label{sec:P.table} Table~\ref{table:fdp} contains some operating characteristics under various settings of the sequential stepdown and stepup procedures, denoted Seq$_{D}$ and Seq$_{U}$, defined in Sections~\ref{sec:DP.procs} and \ref{sec:UP.procs} using step values~\eqref{DP.gen.cons} and \eqref{UP.gen.cons} (both with $\delta_j$ given by \eqref{dn.Holm.DP}), respectively, and  which control $\gfdp\le\alpha=.05$ and $\gfnp\le\beta=.2$. The operating characteristics are the \textit{expected streamwise average sample size} $E_\theta N$ which is the average sample size over the $J$ streams (i.e., $N=\sum_{j=1}^J N_j/J$ if $N_j$ denotes the sample size of the $j$th stream), its standard error SE, and the achieved generalized error rates $\gfdp$ and $\gfnp$. Each operating characteristic estimate is the result of 10,000 Monte Carlo simulated ensembles of $J$ data streams. The parameter values $\gamma_1=\gamma_2=.1$ were used and three states of nature, in terms of the number of true null hypotheses~$H^{(j)}$, are considered for both of the $J=500$ and $J=1000$ scenarios,  with the true $H^{(j)}$ are simulated using $\theta^{(j)}=0$ and the false $H^{(j)}$ with $\theta^{(j)}=1$, representing the ``worst case'' with respect to distinguishability of the null and alternative hypotheses. 

In order to provide a point of reference for these sequential procedures, the performance of comparable fixed sample size stepdown and stepup procedures, denoted by Fix$_{D}$ and Fix$_{U}$, were also estimated. These are the procedures defined in Section~\ref{sec:stats} which use the same step values $\alpha_j$ as Seq$_{D}$ and Seq$_{U}$, respectively. Since these values $\alpha_j$ determine the type~I generalized error rate $\gfdp$, in order to obtain  procedures comparable to the sequential ones, the fixed sample sizes for Fix$_{D}$ and Fix$_{U}$ were chosen as the values yielding the type~II generalized error rate $\gfnp$ most closely matching that of the sequential procedure with the smallest $E_\theta N$ (i.e., the more efficient of Seq$_{D}$ and Seq$_{U}$), whose row is shaded in each scenario in the table.  Regarding this method of choosing the fixed sample size, one might argue that a different comparison would be better since the procedures are quite conservative in their error control, however the fixed sample size procedures are also very conservative, in fact \emph{more} conservative than the sequential procedures since the error probabilities tend to decrease as sample size increases; in this sense the comparison is actually conservative. Because the sample sizes of Fix$_{D}$ and Fix$_{U}$ are fixed, their SE is left blank. The final column of the table shows that savings in $E_\theta N$ of each sequential procedure relative to its fixed sample counterpart.

\begin{table}[h]
\caption{Expected streamwise average sample size $E_\theta N$, its standard error SE, achieved error rates $\gfdp$ and $\gfnp$, and savings in $E_\theta N$ of the sequential (denoted Seq$_{D}$ and Seq$_{U}$) and fixed sample size (denoted Fix$_{D}$ and Fix$_{U}$) procedures described in Section~\ref{sec:P.table} for testing $J$ null hypotheses about the means of Gaussian data streams. The parameter values are $\alpha=.05$, $\beta=.2$, and $\gamma_1=\gamma_2=.1$ and each estimate is the result of 10,000 simulated ensembles of $J$ data streams. The shaded row in each scenario is the procedure with the smallest $E_\theta N$.}
\label{table:fdp}
\begin{center}
\begin{tabular}{c|cccccc}
\# True $H^{(j)}$ & Procedure  & $E_\theta N$     & SE   & $\gfdp(\theta)$ & $\gfnp(\theta)$ & $E_\theta N$ Savings \\\hline\hline
\multicolumn{7}{c}{$J=500$}\\
\multirow{4}{*}{100}   & Seq$_{D}$ & 63.63  & 0.60 & 0.007               & 0.015               & 53\%    \\
      & \cellcolor{lightgray}Seq$_{U}$ & \cellcolor{lightgray}54.17  & \cellcolor{lightgray}0.67 &\cellcolor{lightgray}0.008               & \cellcolor{lightgray}0.012               & \cellcolor{lightgray}55\%    \\
      & Fix$_{D}$ & 136 &      & 0.002               & 0.012               &         \\
      & Fix$_{U}$ & 120 &      & 0.001               & 0.012               &         \\
     & Fix$_{D}'$ & 129 &      & 0.007               & 0.015               &         \\
      & Fix$_{U}'$ & 110 &      & 0.008               & 0.012               &         \\\hline
\multirow{4}{*}{250}   & Seq$_{D}$ & 60.66  & 0.40 & 0.004               & 0.026               & 55\%    \\
      & \cellcolor{lightgray}Seq$_{U}$ & \cellcolor{lightgray}53.39  & \cellcolor{lightgray}0.40 & \cellcolor{lightgray}0.003               & \cellcolor{lightgray}0.016               & \cellcolor{lightgray}58\%    \\
      & Fix$_{D}$ & 135 &      & 0.001               & 0.015               &         \\
      & Fix$_{U}$ & 128 &      & 0.001               & 0.016               &         \\\hline
\multirow{4}{*}{400}   & Seq$_{D}$ & 56.98  & 0.58 & 0.006               & 0.039               & 57\%    \\
      & \cellcolor{lightgray}Seq$_{U}$ & \cellcolor{lightgray}45.97  &\cellcolor{lightgray}0.57 & \cellcolor{lightgray}0.007               & \cellcolor{lightgray}0.022               & \cellcolor{lightgray}65\%    \\
      & Fix$_{D}$ & 134 &      & 0.001               & 0.022               &         \\
      & Fix$_{U}$ & 131 &      & 0.001               & 0.022               & \\\hline\hline       
\multicolumn{7}{c}{$J=1000$}\\
\multirow{4}{*}{250} & Seq$_{D}$ & 67.26 & 0.52 & 0.006 & 0.008 & 54\% \\
    & \cellcolor{lightgray}Seq$_{U}$ & \cellcolor{lightgray}54.93 & \cellcolor{lightgray}0.58 & \cellcolor{lightgray}0.003 & \cellcolor{lightgray}0.009 & \cellcolor{lightgray}56\% \\
    & Fix$_{D}$ & 147      &      & 0.002 & 0.009 &      \\
    & Fix$_{U}$ & 125      &      & 0.001 & 0.009 &      \\\hline
\multirow{4}{*}{500} & Seq$_{D}$ & 65.54 & 0.39 & 0.009 & 0.021 & 50\% \\
    & \cellcolor{lightgray}Seq$_{U}$ & \cellcolor{lightgray}54.06 &\cellcolor{lightgray}0.41 &\cellcolor{lightgray}0.002 &\cellcolor{lightgray}0.027 &\cellcolor{lightgray}54\% \\
    & Fix$_{D}$ & 130      &      & 0.001 & 0.026 &      \\
    & Fix$_{U}$ & 118      &      & 0.001 & 0.026 &      \\\hline
\multirow{4}{*}{750} & Seq$_{D}$ & 63.16 & 0.55 & 0.001 & 0.026 & 54\% \\
    & \cellcolor{lightgray}Seq$_{U}$ &\cellcolor{lightgray}49.72 &\cellcolor{lightgray}0.53 &\cellcolor{lightgray}0.002 &\cellcolor{lightgray}0.023 &\cellcolor{lightgray}61\% \\
    & Fix$_{D}$ & 136      &      & 0.001 & 0.023 &      \\
    & Fix$_{U}$ & 129      &      & 0.001 & 0.023 &     
\end{tabular}
\end{center}
\end{table}

The sequential procedures in Table~\ref{table:fdp} show a dramatic savings in average sample size relative to the fixed sample size procedures, at least 50\% in all cases and as high as 65\%. The sequential procedures also have less conservative error control than their fixed sample size counterparts, most evident in the type~I generalized error rate $\gfdp$ which was not used for ``matching'' the fixed sample procedures as the type~II version was. This less conservative error control is perhaps due to the sequential procedures' smaller average sample size.  Nonetheless, all the procedures still have quite conservative error control relative to the prescribed values of $\alpha=.05$ and $\beta=.2$ even on this highly positively correlated data. Another notable feature of the results in Table~\ref{table:fdp} is that the sequential stepup procedures are slightly but consistently more efficient than the stepdown procedures in each scenario, in terms of minimizing $E_\theta N$. In the next section we will see that the reverse is true in a similar study of procedures controlling \fweI\, and \fweII.

Because of the highly conservative error control of all the procedures in Table~\ref{table:fdp}, but especially the fixed sample size procedures, another type of comparison that may shed light on how much of the efficiency gained by the sequential procedures is due the sequential sampling itself rather than the differing achieved error rates, included in the first scenario of Table~\ref{table:fdp} are two more fixed sample size procedures (denoted Fix$_{D}'$ and Fix$_{U}'$) which match \emph{both} error rates $\gfdp$ and $\gfnp$ of Seq$_{D}$ and Seq$_{U}$, respectively.  These were found by exhaustively searching over values of the fixed streamwise sample size~$N$ and a grid of values for the nominal $\gfdp$ rate $\alpha$ for Fix$_{D}'$ and Fix$_{U}'$. The procedure Fix$_{D}'$ uses $\alpha=.092$ and $N=129$ to match the error rates $\gfdp=.007$ and $\gfnp=.015$ of Seq$_{D}$, and Fix$_{U}'$ uses $\alpha=.112$ and $N=110$ to match $\gfdp=.008$ and $\gfnp=.012$ of Seq$_{U}$. The increase in nominal $\alpha$ required for this matching is roughly a factor of 2, and the decrease in sample size is modest, leaving the sample sizes of Fix$_{D}'$ and Fix$_{U}'$ still substantially larger than their sequential counterparts even though they do not have proven error control at the $\alpha=.05$ level. This suggests the efficiency gains of the sequential procedures relative to the fixed sample size procedures are due more to the sequential sampling  than their less conservative error control.

\subsection{Study of Procedures Controlling $\fweI$ and $\fweII$}
\label{sec:k.table}

Table~\ref{table:fwe} contains the results of a  study similar to Table~\ref{table:fdp} but for procedures controlling \fweI\, and \fweII. In Table~\ref{table:fwe}, Seq$_{D}$ and Seq$_{U}$ denote the stepdown and stepup procedures defined in Sections~\ref{sec:DP.procs} and \ref{sec:UP.procs} using step values~\eqref{cons.Dk} and \eqref{cons.Uk}, respectively, with $\delta_j$ given by \eqref{dn.Uk} for the latter. The parameters $k_1=k_2=25$ were used for the $J=500$ scenario and $k_1=k_2=50$ for the $J=1000$ scenario, and the same prescribed error bounds $\alpha=.05$, $\beta=.2$ were used. The operating characteristics and simulation settings are otherwise the same as the previous section. As there, the stepdown and stepup fixed sample size procedures Fix$_{D}$ and Fix$_{U}$ are those defined in Section~\ref{sec:stats} which use the same step values $\alpha_j$ as Seq$_{D}$ and Seq$_{U}$, respectively, and the fixed sample sizes of these procedures was chosen to match their type~II generalized error rate \fweII\, as closely as possible to the sequential procedure with the smallest $E_\theta N$, whose row is shaded in the table in each scenario.

\begin{table}[h]
\caption{Expected streamwise average sample size $E_\theta N$, its standard error SE, achieved error rates \fweI\, and \fweII, and the savings in $E_\theta N$  of the sequential (denoted Seq$_{D}$ and Seq$_{U}$) and fixed sample size (denoted Fix$_{D}$ and Fix$_{U}$) procedures described in Section~\ref{sec:k.table} for testing $J$ null hypotheses about the means of Gaussian data streams. The parameter values are $\alpha=.05$ and $\beta=.2$ and each estimate is the result of 10,000 simulated ensembles of $J$ data streams. The shaded row in each scenario is the procedure with the smallest $E_\theta N$.}
\label{table:fwe}
\begin{center}
\begin{tabular}{c|cccccc}
\# True $H^{(j)}$ & Procedure  & $E_\theta N$    & SE   & \fweI$(\theta)$ & \fweII$(\theta)$ & $E_\theta N$ Savings \\\hline\hline
\multicolumn{7}{c}{$J=500$, $k_1=k_2=25$}\\
\multirow{4}{*}{100}   & \cellcolor{lightgray}Seq$_{D}$ & \cellcolor{lightgray}38.39 &\cellcolor{lightgray}0.48 &\cellcolor{lightgray}0.020     & \cellcolor{lightgray}0.039     &\cellcolor{lightgray}49\%    \\
      & Seq$_{U}$ & 44.91 & 0.59 & 0.009    & 0.034     & 54\%    \\
      & Fix$_{D}$ & 75    &      & 0.023    & 0.039     &         \\
      & Fix$_{U}$ & 97    &      & 0.002    & 0.040      &         \\
      & Fix$_{D}'$ & 77    &      & 0.020    & 0.039     &         \\
      & Fix$_{U}'$ & 95    &      & 0.009    & 0.034      &         \\\hline
\multirow{4}{*}{250}   & \cellcolor{lightgray}Seq$_{D}$ & \cellcolor{lightgray}36.81 &\cellcolor{lightgray}0.32 & \cellcolor{lightgray}0.017    &\cellcolor{lightgray}0.047     &\cellcolor{lightgray}57\%    \\
      & Seq$_{U}$ & 43.32 & 0.38 & 0.011    & 0.041     & 55\%    \\
      & Fix$_{D}$ & 86    &      & 0.005    & 0.047     &         \\
      & Fix$_{U}$ & 97    &      & 0.001    & 0.046     &         \\\hline
\multirow{4}{*}{400}   & \cellcolor{lightgray}Seq$_{D}$ &\cellcolor{lightgray}32.12 & \cellcolor{lightgray}0.46 & \cellcolor{lightgray}0.007    & \cellcolor{lightgray}0.067     & \cellcolor{lightgray}60\%    \\
      & Seq$_{U}$ & 38.17 & 0.53 & 0.009    & 0.065     & 57\%    \\
      & Fix$_{D}$ & 80    &      & 0.030    & 0.066     &         \\
      & Fix$_{U}$ & 89    &      & 0.001    & 0.066     &\\\hline\hline
      \multicolumn{7}{c}{$J=1000$, $k_1=k_2=50$}\\
\multirow{4}{*}{250}   & \cellcolor{lightgray}Seq$_{D}$ & \cellcolor{lightgray}37.45 &\cellcolor{lightgray}0.42 &\cellcolor{lightgray}0.015    &\cellcolor{lightgray}0.033     &\cellcolor{lightgray}58\%    \\
      & Seq$_{U}$ & 44.07 & 0.51 & 0.005    & 0.042     & 56\%    \\
      & Fix$_{D}$ & 89       &      & 0.009    & 0.034     &         \\
      & Fix$_{U}$ & 100      &      & 0.002    & 0.034     &         \\\hline
\multirow{4}{*}{500}   & \cellcolor{lightgray}Seq$_{D}$ & \cellcolor{lightgray}36.73 & \cellcolor{lightgray}0.31 & \cellcolor{lightgray}0.012    & \cellcolor{lightgray}0.050      & \cellcolor{lightgray}57\%    \\
      & Seq$_{U}$ & 42.46 & 0.38 & 0.008    & 0.044     & 56\%    \\
      & Fix$_{D}$ & 86       &      & 0.005    & 0.051     &         \\
      & Fix$_{U}$ & 96       &      & 0.001    & 0.049     &         \\\hline
\multirow{4}{*}{750}   & \cellcolor{lightgray}Seq$_{D}$ &\cellcolor{lightgray}33.27 & \cellcolor{lightgray}0.41 & \cellcolor{lightgray}0.012    & \cellcolor{lightgray}0.065     & \cellcolor{lightgray}59\%    \\
      & Seq$_{U}$ &39.93 &0.46 & 0.006    & 0.040      & 57\%    \\
      & Fix$_{D}$ & 82       &      & 0.003    & 0.063     &         \\
      & Fix$_{U}$ & 92       &      & 0.001    & 0.064     &        
\end{tabular}
\end{center}
\end{table}

Similar to the results in Table~\ref{table:fdp}, the sequential procedures in Table~\ref{table:fwe} show a substantial savings of roughly 50\% to 60\% in average sample size relative to the fixed sample size procedures, and less conservative error control than their fixed sample size counterparts, most evident in the type~I generalized error rate \fweI\, which was not used for ``matching'' the fixed sample procedures as the type~II version was. All the procedures have quite conservative error control relative to the prescribed values of $\alpha=.05$ and $\beta=.2$. Unlike Table~\ref{table:fdp}, the sequential stepdown procedures in Table~\ref{table:fwe} were more efficient than the stepup procedures in terms of smaller $E_\theta N$. 

Similar to Table~\ref{table:fdp}, the first scenario in Table~\ref{table:fwe} also includes fixed sample size procedures Fix$_{D}'$ and Fix$_{U}'$ whose values of streamwise sample size~$N$ and nominal $\fweI$ bound $\alpha$ were searched over to find values giving attained $\fweI$ and $\fweII$ equal to those of the sequential procedures Seq$_{D}$ and Seq$_{U}$, respectively. The procedure Fix$_{D}'$ uses $\alpha=.048$ and $N=77$ to match the error rates $\fweI=.020$ and $\fweII=.039$ of Seq$_{D}$, and Fix$_{U}'$ uses $\alpha=.080$ and $N=95$ to match $\fweI=.009$ and $\fweII=.034$ of Seq$_{U}$.  Whereas Fix$_{U}'$ uses a slightly smaller sample size (and larger $\alpha$) than Fix$_{U}$ because the latter is more conservative than Seq$_{U}$ in terms of error rates, Fix$_{D}'$ uses a slightly larger sample size (and smaller $\alpha$) than Fix$_{D}$ because the latter is actually less conservative than Seq$_{D}$. In any case, the change in sample size of these modified fixed sample procedures is slight and the fixed sample sizes remain substantially larger than their sequential counterparts, indicating that the increased efficiency is due to the sequential sampling rather than differing achieved error rates, as in Table~\ref{table:fdp}.

\subsection{Composite Hypotheses: Student's $t$-tests}\label{sec:t}

In this section we consider a setting similar to the Gaussian mean testing problem \eqref{z.test.hyp} of the previous sections but drop the assumption of known variance $\sigma^2$, making both the null and alternative in \eqref{z.test.hyp} composite hypotheses.  First we briefly describe a sequential approach to this Student's $t$-test problem and then give the results of a simulation study in a similar setting to Section~\ref{sec:k.table} for $\fweI$ and $\fweII$ control.

Suppose that the data $X_1^{(j)},X_2^{(j)},\ldots$ from a certain data stream are i.i.d.\ Gaussian data with  mean $\mu$ and variance $\sigma^2$, both unknown, and it is desired to test the null hypothesis~$\mu\le 0$ versus the alternative $\mu\ge \delta$, for some given $\delta>0$. Formally, this is a special case of the setup in Section~\ref{sec:setup} by taking $\theta^{(j)}=(\mu,\sigma)^T$, $\Theta^{(j)} =\mathbb{R}\times (0,\infty)$, $H^{(j)} = \{(\mu,\sigma)^T\in\Theta^{(j)}: \mu\le 0 \}$, and $G^{(j)} = \{(\mu,\sigma)^T\in\Theta^{(j)}: \mu\ge\delta \}$.  \citet[][Section~3.2]{Bartroff14b} suggest sequential log generalized likelihood ratio (GLR) statistics for a general class of composite hypotheses when the data is from an exponential family, including this $t$-test setting for which  the sequential log GLR statistic is \citep[see][p.~106]{Bartroff06b}
\begin{gather}\label{t.glr}
\Lambda^{(j)}(n)=\begin{cases}
+\sqrt{2n\Lambda_H(n)},&\mbox{if $\overline{X}_n^{(j)}\ge \delta/2$}\\
-\sqrt{2n\Lambda_G(n)},&\mbox{otherwise,}\end{cases}\\
\mq{where}\Lambda_H(n)=\frac{n}{2}\log\left[1+\left(\frac{\overline{X}_n^{(j)}}{\what{\sigma}_n}\right)^2\right], \quad \Lambda_G(n)=\frac{n}{2}\log\left[1+\left(\frac{\overline{X}_n^{(j)}-\delta}{\what{\sigma}_n}\right)^2\right],\nonumber
\end{gather}
and $\overline{X}_n^{(j)}$ and $\what{\sigma}_n^2$ are the usual MLE estimates of $\mu$ and $\sigma^2$, respectively, based on $X_1^{(j)},\ldots, X_n^{(j)}$. \citet[][Lemma~3.1]{Bartroff14b} also give formulas for certain upper bounds on the probabilities in \eqref{typeI}-\eqref{typeII} involving only properties of the standard normal distribution, allowing critical values $\{A_w^{(j)},B_w^{(j)}\}_{w\in[J]}$ to be computed satisfying \eqref{typeI}-\eqref{typeII} for given step values $\{\alpha_w,\beta_w\}_{w\in[J]}$ by either recursive numerical integration or Monte Carlo simulation of standard normal variates.

Table~\ref{table:t} contains the results of a  study similar to Table~\ref{table:fwe} but for sequential and fixed sample size $t$-tests. In Table~\ref{table:t}, Seq$_{D}$ and Seq$_{U}$ denote the stepdown and stepup procedures defined in Sections~\ref{sec:DP.procs} and \ref{sec:UP.procs} using step values~\eqref{cons.Dk} and \eqref{cons.Uk}, respectively, with $\delta_j$ given by \eqref{dn.Uk} for the latter. The sequential procedures use the statistics~\eqref{t.glr} with critical values computed by Monte Carlo as described in the previous paragraph. The stepdown and stepup fixed sample size procedures Fix$_{D}$ and Fix$_{U}$ are those defined in the first paragraph of Section~\ref{sec:stats} with $p$-values for sample size $n$ computed in the standard way as $1-T_{n-1}(\overline{X}_n^{(j)}\sqrt{n-1}/\what{\sigma}_n)$, where $T_{n-1}(\cdot)$ denotes the c.d.f.\ of the Student's $t$ distribution with $n-1$ degrees of freedom, and which use the same step values $\alpha_j$ as Seq$_{D}$ and Seq$_{U}$, respectively.  As in Section~\ref{sec:k.table},  the fixed sample sizes of these procedures was chosen to match their type~II generalized error rate \fweII\, as closely as possible to the sequential procedure with the smallest $E_\theta N$, whose row is shaded in the table in each scenario. To give a view of the procedures' performance under a different dependency structure for the Gaussian data streams, unlike the previous sections they were simulated not as highly correlated but rather nearly independent with correlation coefficient $.05$ replacing $.95$ in \eqref{covar}. The data was simulated using the same value $\sigma=2$ as above, but not assumed to be known.

\begin{table}[h]
\caption{Expected streamwise average sample size $E_\theta N$, its standard error SE, achieved error rates \fweI\, and \fweII, and the savings in $E_\theta N$  of the sequential (denoted Seq$_{D}$ and Seq$_{U}$) and fixed sample size (denoted Fix$_{D}$ and Fix$_{U}$) procedures described in Section~\ref{sec:t} for testing $J$ null hypotheses about the means of Gaussian data streams with unknown variances. The parameter values are $\alpha=.05$ and $\beta=.2$ and each estimate is the result of 10,000 simulated ensembles of $J$ data streams. The shaded row in each scenario is the procedure with the smallest $E_\theta N$.}
\label{table:t}
\begin{center}
\begin{tabular}{c|cccccc}
\# True $H^{(j)}$ & Procedure  & $E_\theta N$    & SE   & \fweI$(\theta)$ & \fweII$(\theta)$ & $E_\theta N$ Savings \\\hline\hline
\multicolumn{7}{c}{$J=500$, $k_1=k_2=25$}\\
\multirow{4}{*}{100}   & \cellcolor{lightgray}Seq$_{D}$ & \cellcolor{lightgray}40.22 &\cellcolor{lightgray}0.09 &\cellcolor{lightgray}0.003    & \cellcolor{lightgray}0.053     &\cellcolor{lightgray}48\%    \\
      & Seq$_{U}$ & 47.58 & 0.11 & 0.006    & 0.021     & 54\%    \\
      & Fix$_{D}$ & 77    &      & 0.003    & 0.053     &         \\
      & Fix$_{U}$ & 103    &      & 0.002    &0.053      &         \\\hline
\multirow{4}{*}{250}   & \cellcolor{lightgray}Seq$_{D}$ & \cellcolor{lightgray}38.79 &\cellcolor{lightgray}0.04 & \cellcolor{lightgray}0.018    &\cellcolor{lightgray}0.060     &\cellcolor{lightgray}56\%    \\
      & Seq$_{U}$ & 44.79 & 0.05 & 0.003    & 0.018     & 56\%    \\
      & Fix$_{D}$ & 89    &      & 0.006    & 0.059     &         \\
      & Fix$_{U}$ & 101    &      & 0.001    & 0.060     &         \\\hline
\multirow{4}{*}{400}   & \cellcolor{lightgray}Seq$_{D}$ &\cellcolor{lightgray}33.62 & \cellcolor{lightgray}0.09 & \cellcolor{lightgray}0.009    & \cellcolor{lightgray}0.059     & \cellcolor{lightgray}60\%    \\
      & Seq$_{U}$ & 37.10 & 0.11 & 0.011    & 0.057     & 60\%    \\
      & Fix$_{D}$ & 85    &      & 0.037    & 0.060     &         \\
      & Fix$_{U}$ & 93    &      & 0.002   & 0.059     &
\end{tabular}
\end{center}
\end{table}

Comparing Table~\ref{table:t} with the first half of Table~\ref{table:fwe}, one sees that the additional task of estimating the unknown variance in the $t$-test setting, plus the near-independence of the data streams, only cause a modest increase in sample size of all the procedures. The relationship between the sequential and fixed sample size procedures is otherwise remarkably similar to that in Table~\ref{table:fwe}, with the stepdown procedure Seq$_{D}$ being slightly more efficient than the stepup procedure Seq$_{U}$ for FWER control, and both being roughly 50-60\% more efficient than the fixed sample size procedures in terms of expected sample size.  Also like  Table~\ref{table:fwe}, all procedures are very conservative in terms of error control, with the sequential procedures tending to be less so (but not uniformly -- see Fix$_{D}$ in the case of 400 true $H^{(j)}$) because of their smaller average sample size. 

\section{Conclusions and Discussion}\label{sec:conc}

We have proposed general and flexible multiple testing procedures for controlling generalized error rates on sequential data whose error control holds regardless of dependence between data streams. We have given both stepdown and stepup procedures for controlling the tail probabilities of FDP and $k$-FWER, as well as their type~II versions, but in the numerical studies of their performance in Section~\ref{sec:sim} in the setting of highly positively correlated Gaussian data streams we found that, in terms of achieving smaller expected sample size,
\begin{itemize}
\item the stepup procedures performed better for controlling FDP, and 
\item the stepdown procedures performed better for controlling $k$-FWER.
\end{itemize} Although this study was limited to the specific setting of testing hypotheses about the means of Gaussian data streams with covariance matrix~\eqref{covar}, these are our working  recommendations for what to use in practice until further study is possible.

The same simulation studies also show the procedures to be highly conservative in the situation considered, in terms of having generalized error rates substantially smaller than the prescribed values $\alpha$ and $\beta$.  However, it is apparent that this is not related to the sequential nature of the procedures proposed here because the fixed sample versions also have this property and even more so. This is not surprising since the error rates tend to decrease as sample size increases and efficient sequential procedures will have smaller expected sample sizes than their fixed sample counterparts. On the other hand, the results of \citet{Lehmann05b} and \citet{Romano06b} (repeated above in Remarks~\ref{rem:opt.Dk}, \ref{rem:opt.UP}, and \ref{rem:opt.Uk}) show that the error bounds are indeed ``sharp'' and cannot be improved without more restrictive assumptions on the joint distribution of the data streams.  However, less conservative error control (or equivalently, more efficiency in terms of smaller expected sample sizes) may be possible by assumptions about (or direct modeling of) this joint distribution, which was not the focus of this paper but may be a fruitful area of future work.

As mentioned above,

The procedures proposed above, as well as those in \citet{Bartroff14c,Bartroff14b} for FDR/FNR and type~I/II FWER control, are all special cases of the generic sequential procedures in Sections~\ref{sec:gen.D} and \ref{sec:gen.U} and all use the same step values as the corresponding fixed sample size procedures: the \citet{Bartroff14c,Bartroff14b} procedures utilize the same  step values as the \citet{Benjamini95} and \citet{Holm79} procedures, respectively, and the  procedures in this paper utilize the step values of \citet{Lehmann05b} and \citet{Romano06b,Romano06}. Thus, the theme that emerges from this body of work is that, with the appropriate care, fixed sample size step values can be used with the suitable sequential test.

\section{Proofs and Auxiliary Results}\label{sec:proofs}

The proofs of the error-control properties of both the stepup and stepdown procedures utilize the following, as well as Lemma~\ref{lem:DPV} that follows. Let 
\begin{align}
W(j,b)&=\left\{\wtilde{\Lambda}^{(j)}(n)\ge b\;\mbox{for some $n$,}\;\wtilde{\Lambda}^{(j)}(n')> a_1\;\mbox{for all $n'<n$}\right\},\label{T>bM}\\
V_\theta(t,b)&=\bigcup_{j_1,\ldots,j_t\in\mT(\theta)}\bigcap_{\ell=1}^t W(j_\ell,b),\label{V.DP}\\
p^{(j)}(b)&=\sup_{\theta^{(j)}\in H^{(j)}} P_{\theta^{(j)}}(W(j,b)),\label{pj.DP}\\
M_\theta(b)&=\sum_{j\in\mT(\theta)}\bm{1}_{W(j,b)}.\label{Mb.def}
\end{align}  The union in \eqref{V.DP} is over all distinct $t$-tuples $j_1,\ldots,j_t\in\mT(\theta)$. The event $W(j,b)$ is that the standardized test statistic associated with the $j$th null hypothesis crosses $b$ from below before crossing $a_1$ from above, and $V_\theta(t,b)$ is the event that there are at least $t$ true null hypotheses for which this occurs. The function $p^{(j)}(b)$ is the ``worst-case'' (with respect to the null) probability of $W(j,b)$ happening, and the random variable $M_\theta(b)$ is the number of true null hypotheses for which this occurs. Note that any test statistic satisfying the assumptions in Section~\ref{sec:stats}, in particular \eqref{typeI.stand}, satisfies $p^{(j)}(b_w)\le  \alpha_w$ for all $j,w\in[J]$. Note also that the events $W(j,\cdot)$ are non-increasing in the sense that, for any $j\in[J]$,
\begin{equation*}
b\le b'\qmq{implies}W(j,b')\subseteq W(j,b).
\end{equation*}
It follows from this property that the events $V_\theta(t,\cdot)$ are non-increasing, and that $M_\theta(\cdot)$ is non-increasing with probability 1. It can also  be easily verified that $V_\theta(\cdot,b)$ are non-increasing. In what follows we will frequently drop the $\theta$ from $V_\theta$, $M_\theta$, and other quantities when it causes no confusion.

The following lemma is an extension to the sequential domain of \citet[][Lemma~3.1]{Lehmann05b}.

\begin{lemma}\label{lem:DPV}
In the testing situation above, fix $\theta\in\Theta$ such that $\mT(\theta)$ is nonempty, let $t=|\mT(\theta)|\in [J]$, $t_0\in[t]$, and  let $0=\zeta_0\le\ldots\le\zeta_{t_0}\le 1$ and $b_1'\ge b_2'\ge\ldots\ge b_{t_0}'$ be any sequences. With $p^{(j)}(b)$ defined by \eqref{pj.DP}, if the test statistics $\{\Lambda^{(j)}(n)\}$ satisfy $p^{(j)}(b_{s}')\le \zeta_{s}$ for all $j\in\mT(\theta), s\in[t_0]$, then with $V_\theta(t,b)$ defined by \eqref{V.DP}, we have
\begin{equation}\label{PV<.DP.lem}
P_\theta\left(\bigcup_{s=1}^{t_0} V_\theta(s,b_{s}') \right)\le t\sum_{s=1}^{t_0}\frac{\zeta_{s}-\zeta_{s-1}}{s}.
\end{equation}
\end{lemma}

\bigskip

\begin{proof}[Proof of Lemma~\ref{lem:DPV}]
Omit $\theta$ from the notation. With $M(b)$ as in \eqref{Mb.def},
\begin{equation}\label{Etau.s}
E(M(b_s'))=\sum_{j\in\mT}P(W(j,b_s'))\le \sum_{j\in\mT} p^{(j)}(b_s')\le t\zeta_s.
\end{equation}
Define the random variable
$$\tau=\begin{cases}
\min\{s\in[t_0]: \bm{1}_{V(s,b_s')}=1\},&\mbox{if $\bm{1}_{V(s,b_s')}=1$ for some $s\in[t_0]$}\\
t+1,&\mbox{otherwise} \\
\end{cases},$$ and let $\pi_s=P(\tau=s)$. Then the left-hand side of \eqref{PV<.DP.lem} is $$P\left(\bigcup_{s=1}^{t_0}\{\tau=s\}\right)=\sum_{s=1}^{t_0}\pi_s$$ by disjointness. For any $t_1\in [t_0]$ we have 
$\sum_{s=1}^{t_1}\tau\bm{1}_{\{\tau=s\}}=\tau\bm{1}_{\{\tau\le t_1\}}\le M(b_{t_1}')$ by definition of $V$ and $W$. Taking expectations and using \eqref{Etau.s} gives $\sum_{s=1}^{t_1}s\pi_s \le t\zeta_{t_1}$. Dividing both sides of this last by $t_1(t_1+1)$ (resp.\ $t_1$) for $t_1=1,\ldots,t_0-1$ (resp.\ $t_0=t_1$) and summing over $t_1$ gives
\begin{equation}\label{sum.tspi}
\sum_{t_1=1}^{t_0-1}\frac{1}{t_1(t_1+1)}\sum_{s=1}^{t_1}s\pi_s + \frac{1}{t_0}\sum_{s=1}^{t_0}s\pi_s \le \sum_{t_1=1}^{t_0-1}\frac{t \zeta_{t_1}}{t_1(t_1+1)}+ \frac{t\zeta_{t_0}}{t_0}.
\end{equation}
The right-hand side of \eqref{sum.tspi} is easily seen to be the right-hand side of \eqref{PV<.DP.lem}, while the left-hand side of \eqref{sum.tspi} simplifies to $\sum_{s=1}^{t_0}\pi_s$ after reversing the order of summation in the first term.
\end{proof}

\subsection{Proofs of Results for Stepdown Procedures in Section~\ref{sec:D}}

\begin{proof}[Proof of Theorem~\ref{thm:DP}] We consider the generic stepdown procedure defined in Section~\ref{sec:gen.D} with step values given by \eqref{DP.gen.cons}. Fix $\theta\in\Theta$ such that $\mT(\theta)$ is nonempty and omit it from the notation. We show that $\gfdp\le\alpha$, the other claim being similar. Our method of proof here, similar to the other theorems, will be to introduce auxiliary random variables (denoted with superscript $^*$) describing the behavior of the specific sequential stepdown procedure under consideration,  and then establish inequalities that these random variables satisfy on the event $F=\{\fdp>\gamma_1\}$.  In other words, the interpretation of all such inequalities should be that they hold with $P(\cdot|F)$-probability~1, where we assume without loss of generality that $P(F)>0$ since otherwise there is nothing to prove.

As the stepdown procedure proceeds according to its definition in Section~\ref{sec:gen.D}, we may think of the FDP as being updated according to its definition~\eqref{FDP.def} after each rejection of a null hypothesis. In particular, the FDP is updated $m_i$ times  during the $i$th stage according to the ordering~\eqref{Hsrej}. Considering sample paths of data streams on which $\fdp>\gamma_1$ (i.e., outcomes in the event $F$), we define the following auxiliary random variables: Let $i^*$ be the earliest stage at which the updated FDP exceeds $\gamma_1$ at some point during the stage, let $m^*$ denote the smallest value of $m$ such that the rejection of $H^{(j(n_{i^*},|\mJ_{i^*}|-m+1))}$ causes the updated FDP to exceed $\gamma_1$, let $r^*=m^*+m_{i^*}$ so that the $r^*$th rejection causes the FDP to first exceed $\gamma_1$, and finally let $s^*$ be the unique integer $s$ satisfying $s-1\le\gamma_1 r^*<s$ so that $\lfloor \gamma_1 r^*\rfloor+1=s^*$.  Also, define these random variables arbitrary (equal to $\infty$, say) off of $F$ so they are defined on the entire sample space. Note that whereas the $r_i$ in the definition of the procedure in Section~\ref{sec:gen.D} are only updated at the end of each stage, the value of $r^*$ may be determined before the end of the corresponding stage and thus will not necessarily equal any of the $r_i$, even for $i=i^*$. Upon the $r^*$th rejection, $s^*$ true null hypotheses have been rejected. In other words, $F\subseteq V(s^*,b_{r^*})$, where 
\begin{equation}\label{DP.randV}
V(s^*,b_{r^*})=\bigcup_{r,s}\left(V(s,b_r)\cap\{r^*=r\}\cap \{s^*=s\}\right).
\end{equation}
and $V(s,b)$ is defined in \eqref{V.DP}. We note that $b_{r^*}$ is merely an auxiliary random variable useful for proving $\gfdp$ control and is not necessarily equal to, say, the stopping boundary appearing in \eqref{cont-samp} for any particular stage of the procedure.

The number~$|\mT|$ of true null hypotheses satisfies
\begin{equation}\label{mT<=}
|\mT|\le J-(r^*-s^*).
\end{equation}
 Also, since $s^*/\gamma_1$ (interpreted as $\infty$ if $\gamma_1=0$) strictly exceeds $r^*$ we have $r^*\le\lceil s^*/\gamma_1\rceil -1$. Denoting $\overline{j}(s,|\mT|,\gamma_1)$ and $D_1(\gamma_1,\{\delta_j\})$ simply by $\overline{j}(s)$ and $D_1$, respectively, these last two bounds imply that 
 \begin{equation}\label{ar<eaS}
r^*\le \overline{j}(s^*).
\end{equation} 

Next we bound the possible values of $s^*$ from above. Because $r^*\le J$ we have $s^*\le\lfloor \gamma_1 J\rfloor+1$, and clearly $s^*\le|\mT|$. Also, by the minimal properties of $i^*$ and $m^*$ it must be that the previous update of the FDP, which took the value $(s^*-1)/(r^*-1)$, did not exceed $\gamma_1$, hence $(s^*-1)/(r^*-1)\le \gamma_1$, or
\begin{equation*}
s^*\le\gamma_1(r^*-1)+1\le \gamma_1(J-|\mT|+s^*-1)+1,
\end{equation*} using \eqref{mT<=}. Simplifying and taking the floor gives
\begin{equation}\label{t.up}
s^*\le \left\lfloor \frac{\gamma_1(J-|\mT|)}{1-\gamma_1}\right\rfloor+1,
\end{equation} and combining these three bounds gives
\begin{equation}\label{s^*<tup}
s^*\le \overline{t}(|\mT|,\gamma_1),
\end{equation}
which we denote simply by $\overline{t}$. 

Let $F_1$ be the event that \eqref{ar<eaS} and \eqref{s^*<tup} hold, and note that $F\subseteq F_1$ by the arguments leading to those inequalities. Then
\begin{multline}\label{PFDP.D1}
\gfdp =P(F)=P(F\cap F_1) \le  P(V(s^*,b_{r^*})\cap F_1)\le \sum_{s'=1}^{\overline{t}} P(V(s',b_{r^*})\cap\{s^*=s'\}\cap F_1)\\
\le \sum_{s'=1}^{\overline{t}} P(V(s',b_{\overline{j}(s^*)})\cap\{s^*=s'\}\cap F_1)\qm{[since \eqref{ar<eaS} holds on $F_1$ and  $V(s',\cdot)$ are non-increasing]}\\
=\sum_{s'=1}^{\overline{t}} P(V(s',b_{\overline{j}(s')})\cap\{s^*=s'\}\cap F_1)\le \sum_{s'=1}^{\overline{t}} P\left(\bigcup_{s=1}^{\overline{t}} V(s,b_{\overline{j}(s)})\cap\{s^*=s'\}\cap F_1\right)\\
\le P\left(\bigcup_{s=1}^{\overline{t}} V(s,b_{\overline{j}(s)}) \cap F_1\right)\le P\left(\bigcup_{s=1}^{\overline{t}} V(s,b_{\overline{j}(s)}) \right).
\end{multline}
The third-to-last inequality in \eqref{PFDP.D1} holds by expanding the event $V(s',b_{\overline{j}(s')})$ into the union event.

Now let $\eps_s=\eps(s,|\mT|,\gamma_1,\{\delta_j\})$ and it is not hard to verify that the $\eps_s$ are non-decreasing in $s$. We have
\begin{equation}\label{DP.PW.upper}
P(W(j,b_{\overline{j}(s)}))\le\alpha_{\overline{j}(s)}=\frac{\alpha \delta_{\overline{j}(s)}}{D_1}=\frac{\alpha \eps_s}{D_1}.
\end{equation} 
Applying Lemma~\ref{lem:DPV} (with $\zeta_s=\alpha \eps_s/D_1$) to the last term in \eqref{PFDP.D1} gives
$$\gfdp\le |\mT| \sum_{s=1}^{\overline{t}} \frac{\alpha_{\overline{j}(s)}-\alpha_{\overline{j}(s-1)}}{s} = \frac{\alpha |\mT|}{D_1}\sum_{s=1}^{\overline{t}} \frac{\eps_s-\eps_{s-1}}{s}=\frac{\alpha S}{D_1}\le\alpha.$$

The proof that $\gfnp\le\beta$ is completely symmetric and there is no conflict or interaction between rejections and acceptances of null hypotheses.

\end{proof}

\bigskip

\begin{proof}[Proof of Theorem~\ref{thm:Dk}]
We consider the generic stepdown procedure defined in Section~\ref{sec:gen.D} with step values given by \eqref{cons.Dk}. We will show that the procedure satisfies $\fweI(\theta) \le\alpha$, the other claim being similar. Fix $\theta\in\Theta$ such that $|\mT(\theta)|\ge k_1$, since otherwise $\fweI(\theta) =0$, and omit $\theta$ from the notation.  As in the proof of Theorem~\ref{thm:DP} we shall introduce auxiliary random variables (denoted with superscript $^*$) describing the behavior of the procedure and then establish inequalities that these random variables satisfy on the event 
$$F=\{\mbox{at least $k_1$ null hypotheses $H^{(j)}$ rejected, $j\in\mathcal{T}$}\}.$$
In other words, the interpretation of all such inequalities should be that they hold with $P(\cdot|F)$-probability~1, where we assume without loss of generality that $P(F)>0$ since otherwise there is nothing to prove.

Considering outcomes on which at least $k_1$ true null hypotheses are rejected (i.e., outcomes in $F$), define the following random variables: Let $i^*$ and $m^*$ be the stage and index, respectively, of the $k_1$th rejected true null hypothesis~$H^{(j(n_{i^*},|\mJ_{i^*}|-m^*+1))}$, and let $r^*=m^*+m_{i^*}$ so that the $r^*$th null hypothesis rejected is the $k_1$th true null hypothesis rejected. As above, define these random variables arbitrarily off of $F$ so they are defined on the entire sample space. Then
\begin{equation}\label{r.bd.Dk}
r^*\le J-|\mT|+k_1
\end{equation} because $r^*$ takes its largest possible value  when  the only remaining hypotheses upon the $r^*$th rejection are the $|\mT|-k_1$ remaining true null hypotheses. As in the proof of Theorem~\ref{thm:DP}, $r^*$ may be determined before the end of stage $i^*$ and should not be confused with the $r_i$ in the definition of the procedure in Section~\ref{sec:gen.D}.

With $M(b)$ as in \eqref{Mb.def}, note that $M(b_{r^*})\ge k_1$ on $F$ since at least $k_1$ standardized test statistics corresponding to true nulls must cross $b_{r^*}$ before crossing the furthest lower boundary~$a_1$, since otherwise a null would be accepted rather than rejected by definition of the stepdown procedure.  Thus $F\subseteq\{M(b_{r^*})\ge k_1\}$ and so
\begin{equation}\label{Dk.Mmono}
M(b_{r^*})\le M(b_{J-|\mT|+k_1})
\end{equation} by \eqref{r.bd.Dk} and the monotonicity of $M(\cdot)$ and the $b_j$. Then, using \eqref{Dk.Mmono} and Markov's inequality,
\begin{multline*}
\fweI= P(F)\le P(M(b_{r^*})\ge k_1)\le P(M(b_{J-|\mT|+k_1})\ge k_1)\le\frac{1}{k_1}E\left( M(b_{J-|\mT|+k_1}) \right)\\
= \frac{1}{k_1}E\left( \sum_{j\in\mT}\bm{1}_{W(j,b_{J-|\mT|+k_1})} \right) =\frac{1}{k_1}\sum_{j\in\mT} P(W(j,b_{J-|\mT|+k_1}))\le \frac{1}{k_1}\sum_{j\in\mT}\alpha_{J-|\mT|+k_1}\\=
\frac{|\mT|\alpha_{J-|\mT|+k_1}}{k_1}=\alpha.
\end{multline*}
\end{proof}

\subsection{Proofs of Results for Stepup Procedures in Section~\ref{sec:U}}

The proofs of both Theorems~\ref{thm:UP} and \ref{thm:Uk} utilize the following lemma.

\begin{lemma}\label{lem:Rs<V.U}
For the generic sequential stepup procedure in Section~\ref{sec:gen.U}, under $\theta\in\Theta$, for any $s\in[J]$ we have
\begin{equation}\label{Rs<V.U}
\{\mbox{exactly $s$ null hypotheses rejected}\;\} \subseteq  V_\theta(t^*,b_s),
\end{equation} the latter defined as
\begin{equation}\label{def.Vt*}
V_\theta(t^*,b_s)=\bigcup_t \left(V_\theta(t,b_s)\cap \{t^*=t\}\right)
\end{equation}
where $V_\theta(t,b)$ is as in \eqref{V.DP} and $t^*$ is the number of true null hypotheses rejected.
\end{lemma}

\begin{proof}
Let $R(s)$ denote the event on the left-hand side of \eqref{Rs<V.U}.
On outcomes in $R(s)$ define the following random variables: Let $i^*$ be the stage at which the $s$th rejection occurs, let $j^*$ be such that $H^{(j^*)}$ is the $s$th rejected null hypothesis, and recall that $t^*$ is the number of true hypotheses rejected. By definition of $H^{(j^*)}$ and by step~\ref{fdrrej-step} of the procedure we have $j^*=j(n_{i^*},|\mJ_{i^*}|-m_{i^*}+1)$, $s=r_{i^*}+m_{i^*}$, and $\wtilde{\Lambda}^{(j^*)}(n_{i^*})\ge b_{r_{i^*}+m_{i^*}}$. If a true hypothesis~$H^{(j')}$, $j'\in\mT$, is rejected then it is rejected at some stage $i'\le i^*$, and $j'=j(n_{i'},\ell)$ for some $\ell\ge |\mJ_{i'}|-m_{i'}+1$. Note that $r_{i'}+m_{i'}\le s$  because if $i'=i^*$ then $r_{i'}+m_{i'}=r_{i^*}+m_{i^*}=s$, and otherwise $i'\le i^*-1$ so $r_{i'}+m_{i'}=r_{i'+1}\le r_{i^*}\le s$.
Then
$$\wtilde{\Lambda}^{(j')}(n_{i'})= \wtilde{\Lambda}^{(j(n_{i'},\ell))}(n_{i'})\ge \wtilde{\Lambda}^{(j(n_{i'},|\mJ_{i'}|-m_{i'}+1))}(n_{i'})\ge b_{r_{i'}+m_{i'}}\ge b_s,$$ using \eqref{fdrmjrej} for the second-to-last inequality. This holds for any rejected true hypothesis, hence there are distinct $j_1,\ldots,j_{t^*}\in\mT$ such that, for each $\ell\in[t^*]$, $H^{(j_\ell)}$ is rejected and 
\begin{equation}\label{L>bs.UP}
\wtilde{\Lambda}^{(j_\ell)}(n)\ge b_s\qm{for some $n$.} 
\end{equation}
If it were that $\wtilde{\Lambda}_{n'}^{(j_\ell)}\le a_1$ for some $n'$ less than the corresponding $n$ in \eqref{L>bs.UP}, then $H^{(j_\ell)}$ would not have been rejected but rather accepted, contradicting our assumption about $H^{(j_\ell)}$. Combining these statements gives that any outcome in $R(s)$ is in $V(t^*,b_s)$. 
\end{proof}

\bigskip

\begin{proof}[Proof of Theorem~\ref{thm:UP}] We consider the generic stepup procedure defined in Section~\ref{sec:gen.U} with step values given by \eqref{UP.gen.cons}. Fix $\theta\in \Theta$ such that $\mT(\theta)$ is nonempty, and omit $\theta$ from the notation. We will show that $\gfdp\le\alpha$, the other claim being similar.  For $s\in[J]$ let $\gamma(s)=\lfloor\gamma_1 s\rfloor+1$ and let $T(s)$ denote the event that at least $\gamma(s)$ true null hypotheses are rejected, and let $R(s)$ be the event on the left-hand side of \eqref{Rs<V.U}. We claim that 
\begin{equation}\label{T<V.UP}
R(s)\cap T(s)\subseteq \begin{cases}
V(\gamma(s)\vee (s+|\mT|-J),b_s),&\mbox{if $\gamma(s)\le |\mT|$,}\\
\emptyset,&\mbox{otherwise.}
\end{cases}
\end{equation} By Lemma~\ref{lem:Rs<V.U} we have $R(s)\cap T(s)\subseteq V(t^*,b_s)\cap T(s)$, and to finish the proof of \eqref{T<V.UP} we show that, on any outcome in the latter event, defined analogously to \eqref{def.Vt*}, $t^*\ge \gamma(s)\vee (s+|\mT|-J)$ if $s$ is such that $\gamma(s)\le |\mT|$ and then use that $V(\cdot,b)$ is non-increasing; the other case of \eqref{T<V.UP} is trivial by the definition of $R(s)$ and $T(s)$. We recall that this and other inequalities involving the random variable $t^*$ should be interpreted as holding with $P(\cdot|V(t^*,b_s)\cap T(s))$-probability~$1$, this event assumed without loss of generality to have positive probability. We know that $t^*\ge \gamma(s)$ by definition of $T(s)$. On the other hand, $t^*$ is equal to the number~$s$ of null hypotheses rejected minus the number of false null hypotheses rejected, the latter bounded above by $J-|\mT|$, hence $t^*\ge s+|\mT|-J$.

With \eqref{T<V.UP} established we have
\begin{multline}
\gfdp=\bigcup_{1\le s\le J} R(s)\cap T(s) \subseteq \bigcup_{1\le s\le J,\; \gamma(s)\le|\mT|} V(\gamma(s)\vee (s+|\mT|-J),b_{s})\\
= \bigcup_{|\mT|-J+1\le s\le |\mT|,\; \gamma(J+s-|\mT|)\le |\mT|} V(\gamma(J+s-|\mT|)\vee s,b_{J+s-|\mT|}).\label{PF<U.UP}
\end{multline} For $s$ in the range of the union in \eqref{PF<U.UP}, let $\sigma(s)=\gamma(J+s-|\mT|)\vee s$, which is  a non-decreasing sequence of consecutive integers taking the values $1,2,\ldots,s_1$ for some $s_1\le|\mT|$ by virtue of the restrictions in \eqref{PF<U.UP}.  For $s \in [s_1]$ let $\sigma^{-1}(s)=\max\{s': \sigma(s')=s\}$. If $\sigma(s)=\sigma(s+1)$ then using the non-increasing property of $V(s,\cdot)$ we have
\begin{equation}\label{V.tele.UP}
V(\sigma(s),b_{J+s-|\mT|})\cup V(\sigma(s+1),b_{J+s+1-|\mT|})\subseteq V(\sigma(s+1),b_{J+s+1-|\mT|}).
\end{equation}
 By collapsing terms in \eqref{PF<U.UP} according to \eqref{V.tele.UP}, we have that the union in \eqref{PF<U.UP} is contained in
 \begin{equation*}
\bigcup_{s=1}^{s_1} V(s,b_{J+\sigma^{-1}(s)-|\mT|}).
\end{equation*} Denote $S_2(|\mT|,\gamma_1,\{\delta_j\})$ and $D_2(\gamma_1,\{\delta_j\})$ by $S_2$ and $D_2$, respectively.
Applying Lemma~\ref{lem:DPV} to this last with $\zeta_s=\alpha_{J+\sigma^{-1}(s)-|\mT|}$ and $b_s'=b_{J+\sigma^{-1}(s)-|\mT|}$, and recalling that $s_1\le|\mT|$, we have
\begin{multline}
\frac{D_2}{|\mT|\alpha}\cdot P(\fdp >\gamma_1)\le \frac{D_2}{|\mT|\alpha}\cdot P\left(\bigcup_{s=1}^{s_1} V(s,b_{J+\sigma^{-1}(s)-|\mT|})\right)\\
\le \frac{D_2}{|\mT|\alpha}\cdot  |\mT|\left( \alpha_{J+\sigma^{-1}(1)-|\mT|} +\sum_{1<s\le s_1} \frac{\alpha_{J+\sigma^{-1}(s)-|\mT|} - \alpha_{J+\sigma^{-1}(s-1)-|\mT|}}{s} \right)\\
=  \delta_{J+\sigma^{-1}(1)-|\mT|} +\sum_{1<s\le s_1} \frac{\delta_{J+\sigma^{-1}(s)-|\mT|} - \delta_{J+\sigma^{-1}(s-1)-|\mT|}}{s}.\label{sum.d.UP}
\end{multline} We claim that \eqref{sum.d.UP} is equal to
\begin{equation}\label{S2.UP}
S_2/|\mT| = \delta_1+ \sum_{|\mT|-J+1<s\le |\mT|,\; |\mT|\ge \lfloor \gamma_1(J-|\mT|+s)\rfloor+1} \frac{\delta_{J-|\mT|+s}-\delta_{J-|\mT|+s-1}}{s\vee (\lfloor \gamma_1 (J-|\mT|+s)\rfloor +1)},
\end{equation}
which would complete the proof since $S_2\le D_2$. Note that the denominator in \eqref{S2.UP} is $\sigma(s)$.  If $\sigma^{-1}(1)=|\mT|-J+1$ then the first term in both \eqref{sum.d.UP} and \eqref{S2.UP} is $\delta_1$. Otherwise, $\sigma^{-1}(1)=s_2>|\mT|-J+1$, and the first $J+s_2-|\mT|$ summands in \eqref{S2.UP} are
\begin{multline*}
\delta_1+\frac{\delta_2-\delta_1}{\sigma(|\mT|-J+2)}+\ldots +\frac{\delta_{J+s_2-|\mT|}-\delta_{J+s_2-|\mT|-1}}{\sigma(s_2)} \\
= \delta_1+\frac{\delta_2-\delta_1}{1}+\ldots +\frac{\delta_{J+s_2-|\mT|}-\delta_{J+s_2-|\mT|-1}}{1}=\delta_{J+s_2-|\mT|},
\end{multline*} which is the first summand in \eqref{sum.d.UP} in this case.  Proceeding in this way one may verify the claim term by term, completing the proof. 
\end{proof}

\bigskip

\begin{proof}[Proof of Theorem~\ref{thm:Uk}] We consider the generic stepup procedure defined in Section~\ref{sec:gen.U} with step values given by \eqref{cons.Uk}. We will show that the procedure satisfies $\fweI(\theta) \le\alpha$, the other claim being similar. Fix $\theta\in\Theta$ such that $|\mT(\theta)|\ge k_1$, since otherwise $\fweI(\theta) =0$, and omit $\theta$ from the notation. For $s\in[J]$ let $T(s)$ denote the event that at least $k_1$ true null hypotheses are rejected and let $R(s)$ be the event on the left-hand side of \eqref{Rs<V.U}. We claim that 
\begin{equation}\label{T<V.Uk}
R(s)\cap T(s)\subseteq V(k_1\vee (s+|\mT|-J),b_s)\qmq{for all}s\in[J].
\end{equation} By Lemma~\ref{lem:Rs<V.U} we have $R(s)\cap T(s)\subseteq V(t^*,b_s)\cap  T(s)$, the latter defined analogously to \eqref{def.Vt*}, and to finish the proof of \eqref{T<V.Uk} we show that, on any outcome in the latter event, $t^*\ge k_1\vee (s+|\mT|-J)$ and use that $V(\cdot,b)$ is non-increasing. We recall that this and other inequalities involving the random variable $t^*$ should be interpreted as holding with $P(\cdot|V(t^*,b_s)\cap T(s))$-probability $1$, this event assumed without loss of generality to have positive probability. We know that $t^*\ge k_1$ by definition of $T(s)$. On the other hand, $t^*$ is equal to the number~$s$ of null hypotheses rejected minus the number of false null hypotheses rejected, the latter bounded above by $J-|\mT|$, hence $t^*\ge s+|\mT|-J$.

With \eqref{T<V.Uk}  established we have
\begin{align}
\bigcup_{k_1\le s\le J} R(s)\cap T(s)&\subseteq \bigcup_{k_1\le s\le J} V(k_1\vee (s+|\mT|-J),b_s)\label{UT=UV.Uk}\\
 &= \left\{\bigcup_{k_1\le s\le J-|\mT|+k_1} V(k_1,b_s) \right\}\cup  \left\{ \bigcup_{J-|\mT|+k_1< s\le J} V(s+|\mT|-J,b_s)\right\}\nonumber\\
 &\subseteq V(k_1,b_{J-|\mT|+k_1}) \cup  \left\{ \bigcup_{J-|\mT|+k_1< s\le J} V(s+|\mT|-J,b_s)\right\}\label{UV=V.Uk}\\
 &= \bigcup_{J-|\mT|+k_1\le s\le J} V(s+|\mT|-J,b_s) = \bigcup_{k_1\le s\le |\mT|} V(s,b_{|\mT|-J+s}),\label{UV'.Uk}
\end{align} where the inclusion in \eqref{UV=V.Uk} follows from the facts that $b_s\ge b_{J-|\mT|+k_1}$ for $s\le J-|\mT|+k_1$ and the $V(k_1,\cdot)$ are non-increasing. The $\fweI$ is the probability of the event on the left-hand side of \eqref{UT=UV.Uk}, and applying Lemma~\ref{lem:DPV} to the last union in \eqref{UV'.Uk} with $t_0=|\mT|$, $\zeta_0=\ldots=\zeta_{k_1-1}=0$, $\zeta_s=\alpha_{|\mT|-J+s}$ for $k_1\le s\le |\mT|$, $b_1'=\ldots=b_{k_1-1}'=\infty$, and $b_s'=b_{|\mT|-J+s}$ for $k_1\le s\le |\mT|$, we have
\begin{multline*}
\fweI\le |\mT|\left(\frac{\alpha_{|\mT|-J+k_1}}{k_1}+\sum_{k_1<s\le |\mT|} \frac{\alpha_{|\mT|-J+s}-\alpha_{|\mT|-J+s-1}}{s}\right)\\
=\left(\frac{\alpha}{D_3(k_1,\{\delta_j\})}\right)S_3(k_1,|\mT|,\{\delta_j\})\le \alpha.
\end{multline*}
\end{proof}

\subsection{Proof of Theorem~\ref{thm:simple}}
\renewenvironment{proof}{}{\qed}
\begin{proof}
We verify the first parts of \eqref{sum<=1} and \eqref{fdraH=aS}; the other parts are similar. We have
$$0\le\alpha_w+\wtilde{\beta}_w=\alpha_w+ \frac{\beta_1(1-\alpha_w)}{1-\alpha_1} -1+1 =\frac{-(1-\alpha_w)(1-\alpha_1-\beta_1)}{1-\alpha_1} +1\le 1,$$ using $\alpha_1+\beta_1\le 1$ for the last inequality. It is simple algebra to check that $B_w^{(j)}$ in \eqref{fdrAsBs} can be written as $B_W(\alpha_w,\wtilde{\beta}_w)$, and $A_1^{(j)}$ in \eqref{fdrAsBs} can be written as $A_W(\alpha_w,\wtilde{\beta}_w)$ for any $w\in[J]$. Then 
\begin{align*}
p_w^{(j)}&=P_{h^{(j)}}(\Lambda^{(j)}(n)\ge B_w^{(j)}\mbox{ some $n$, }\Lambda^{(j)}(n')>A_1^{(j)}\;\mbox{all $n'<n$})\\
&=P_{h^{(j)}}(\Lambda^{(j)}(n)\ge B_W(\alpha_w,\wtilde{\beta}_w)\;\mbox{some $n$,}\; \Lambda^{(j)}(n')>A_W(\alpha_w,\wtilde{\beta}_w)\;\mbox{all $n'<n$}) \\
&=\alpha_W^{(j)}(\alpha_w,\wtilde{\beta}_w),
\end{align*}
by definition of $\alpha_W^{(j)}$.
\end{proof}


\def\cprime{$'$}

\end{document}